\documentclass[11pt]{article}
\usepackage{times}
\usepackage[margin=1in]{geometry}
\usepackage{fullpage}
\usepackage[reqno,tbtags]{amsmath}
\usepackage{amsthm,amssymb}
\usepackage[dvipdfmx]{graphicx}
\usepackage{bmpsize}
\usepackage{psfrag,subfigure}
\usepackage{afterpage} 
\usepackage{amsfonts}
\usepackage{color}
\usepackage{enumerate}
\usepackage{multirow}
\usepackage[numbers,sort&compress]{natbib}
\usepackage[dvipsnames,usenames,table]{xcolor}
\usepackage[backref,bookmarks=true,colorlinks=true,urlcolor=blue,linkcolor=blue,citecolor=OliveGreen]{hyperref}
\usepackage{cleveref}
\usepackage{float}

\newtheorem{theorem}{Theorem}[section]
\newtheorem{claim}{Claim}
\newtheorem{lemma}[theorem]{Lemma}

\newtheorem{thm}[theorem]{Theorem}
\newtheorem{cor}[theorem]{Corollary}

\theoremstyle{definition}

\theoremstyle{remark}

\newenvironment{numberedlemma}[1]{%
\renewcommand{\thetheorem}{#1}%
\begin{lemma}}{\end{lemma}\addtocounter{theorem}{-1}}

\newcommand{\opt}{\textsc{OPT}}

\numberwithin{equation}{section}

\newcommand{\comment}[1]{}
\newcommand{\pr}[1]{\mathbb{P}\!\left(#1\right)}
\newcommand{\EE}[1]{\mathbb{E}\!\left[#1\right]}

\newcommand{\ba}{\mathbf a}

\newcommand{\bA}{\mathbf A}

\newcommand{\Paths}{\mathcal{P}}

\newcommand{\inn}{\mathrm{in}}
\newcommand{\outt}{\mathrm{out}}

\newcommand{\dmax}{d_{\max}}

\newcommand{\binaryTree}{\mathbf{T}}
\newcommand{\Trees}[1]{\mathcal{T}(#1)}

\newcommand{\stgeq}{ \overset {\rm st} \geq  }

\newcommand{\favj}{\textsc{Fav}_j}

\newcommand{\integers}{\mathbb{Z}}
\newcommand{\reals}{\mathbb{R}}
\newcommand{\lmax}{l_{\max}}

\newcommand{\Pathj}{\Pi_j}
\newcommand{\realPaths}{\mathcal{P}}
\newcommand{\Anc}{C}
\newcommand{\Btl}[1][t]{\tilde{B}_{#1,l}}
\newcommand{\layerl}[1][l]{\Gamma_{#1}}
\newcommand{\DD}{{\mathcal D}}
\newcommand{\di}{\text{dep}_i}

\begin{document}
\title{
Stability of Service under Time-of-Use Pricing}
\author{Shuchi Chawla\thanks{University of Wisconsin-Madison. \tt{shuchi@cs.wisc.edu}} \and Nikhil
  R. Devanur\thanks{Microsoft Research,
    Redmond. \tt{nikdev@microsoft.com}} \and Alexander E. Holroyd\thanks{Microsoft Research, Redmond. \tt{holroyd@microsoft.com}}  \and
  Anna R. Karlin\thanks{University of Washington,
    Seattle. \tt{karlin@cs.washington.edu}} \and James
  Martin\thanks{Oxford University. \tt{martin@stats.ox.ac.uk}}
  \and Balasubramanian Sivan\thanks{Google Research. \tt{balusivan@google.com}}}
\date{}
\maketitle
\begin{abstract}
We consider ``time-of-use" pricing as a technique for matching supply and demand of temporal resources with the goal of maximizing social welfare.
Relevant examples include energy, computing resources on a cloud computing platform, and charging stations for electric vehicles, among many others.
A client/job in this setting has a window of time during which he needs service, and a particular value for obtaining it. 
We assume a stochastic model for demand, where each job materializes with some probability via an independent Bernoulli trial.
Given a per-time-unit pricing of resources, any realized job will first try
to get served by the cheapest available resource in its window and, failing that, will try to find service at the next cheapest available resource, and so on. Thus, the natural stochastic fluctuations in demand  have the potential to lead to cascading overload events. 
Our main result shows that setting prices so as to optimally handle
the {\em expected} demand works well: with high probability, when the actual  demand is instantiated, the system is stable and the expected value of the jobs served is very close to that of the optimal offline algorithm.

\end{abstract}
\thispagestyle{empty}
\newpage
\setcounter{page}{1}

\section{Introduction}

For many commodities of a temporal nature, demand and supply fluctuate
stochastically over time. Demand for electricity changes over the course of a
day as well as across different days of the week---home owners use more
electricity during evenings and weekends, offices use more during normal working
hours, and energy usage at factories can follow an altogether different cycle
depending on workloads. On the other hand, supply from sources of renewable
energy depends on weather conditions and can also vary significantly over time.
Likewise, demand for computing resources on a cloud computing platform varies
over time depending on users' workloads. Supply also varies stochastically,
depending on scheduled and unscheduled downtime for servers and other kinds of
outages. In this paper, we explore the effectiveness of ``time-of-use" pricing
as a method for efficiently and effectively matching supply and demand in such settings. 

\paragraph{Online matching of temporal resources:}
Suppose that there are $B_t$ units of resource available at time $t$ and each potential client, a.k.a. ``job'', $j$ has a window of time during which it would like to obtain ``service'', say, a unit allocation of the resource. Job $j$ obtains a value of $v_j$ from getting serviced at any time in its window. (See figure below.) We consider the following model of job arrival that is a hybrid of stochastic and adversarial models: job $j$ is realized with probability $q_j$ via an independent Bernoulli trial. Jobs arrive online in the system in an adversarial order that can depend on the set of realized jobs.


\begin{figure}[htbp]
\begin{center}
\includegraphics[width=0.6\textwidth]{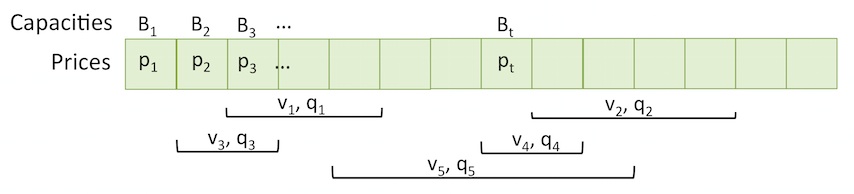}
\end{center}
\end{figure}

How should the supplier allocate the available resources to jobs so as
to maximize the total expected value\footnote{ Also called the efficiency or social welfare of the system.} of the jobs that are served? Perhaps the most natural approach is to use the stochastic information about demand
to price the available resources on a per-time-unit basis. Such
\textbf{time-of-use pricing} is an effective way for the supplier of a temporal commodity to balance supply and demand. During lean supply periods, advertising a high price suppresses demand, whereas during times of excess supply, advertising a low price encourages higher demand. 
Moreover, allowing each client to be a ``price-taker'', that is, making sure
that each client is allocated the cheapest
available resource that meets his requirements,\footnote{ A resource
at time $t$ priced at $p_t$ meets job $j$'s requirement if $p_t \le v_j$ and time slot $t$ is
within job $j$'s window.} trivially guarantees
that clients will be truthful about all of their parameters: there is no advantage to misreporting one's value or service window.

For a single time period in isolation, determining the right price to set is a \emph{newsvendor} problem~\cite{stevenson-om}.
The optimal solution is to set the price so that the system is slightly overprovisioned with the expected supply matching the expected demand plus a small reserve. 
Even for this setting, if jobs arrive in adversarial order, $\Omega ( \epsilon^{-2} \ln (\epsilon ^{-1}))$ units of resource are needed to guarantee
that the expected value of the jobs served is at least $1-\epsilon$ times
 the expected value of the jobs scheduled by the optimal offline algorithm.

For the general case, as a thought experiment, suppose that we only needed to satisfy supply constraints in expectation in every time period. This entails solving an ``expected LP'', which yields a set of prices, one for each period,
and automatically matches potential jobs with a cheapest slot in their window. But how well does such a system work under the natural stochastic fluctuations that will necessarily occur? The concern is that, because of variability in the realized
demand relative to the expected demand, a client may show up and find
that the cheapest slot in his window has already been allocated. This will cause him to try to take the next cheapest slot and so on. Such ``overload'' events, that is, events where demand exceeds supply causing excess demand to be forwarded, are positively correlated across time slots and can exhibit cascading behavior. 

Our main theorem is that such time-of-use pricing  works well with high probability: Suppose that $B_t$ resources are available at each time $t$, where $B_t = \Omega ( \epsilon^{-2} \ln (\epsilon^{-1}))$. (We call this the ``large market assumption.") Given the model of job arrivals described above, there is a set of prices $(p_t)$ such that if (a) realized jobs arrive online in an adversarial order,
and (b) upon arrival, each job grabs his favorite available resource given
the prices,  the expected value of the jobs served is at least $1-\epsilon$ times the expected value of the jobs scheduled by the optimal offline algorithm.

Thus, despite the complex interaction between demand for time slots due to the forwarding of unmet demand and the adversarial arrival order of realized jobs, we can guarantee near-optimal expected
performance without increasing capacity over what would be needed in a single time-unit setting.

\paragraph{Key ideas in the proof.} The prices we set induce a {\em forwarding graph}: the nodes
are time slots and an edge from time slot $t$ to time slot $t'$ means
that $p_t \le p_{t'}$ and some job might try to grab a resource
at time $t'$ immediately after failing to find an available resource at time $t$. 

What properties of the forwarding graph determine whether or not
overload cascades are likely? 
Perhaps unsurprisingly, the maximum in-degree of a node in the
forwarding graph is key. Suppose, for example, that one time slot $t$ has very high in-degree,
meaning that it may receive forwarded jobs from many other time slots. Even
if each of the latter time slots has a low probability of forwarding a job, the total expected number of jobs forwarded to $t$ may be high, and may therefore lead to a high probability of overload at $t$. If all of the highest
value jobs happen to have $t$ as the only slot in their window, this could wreak havoc on our social welfare bounds.

 What {\em is} perhaps somewhat surprising though is that maximum in-degree is the {\em only} relevant graph parameter. In particular, the size of the graph does not play a role.
Showing that our theorem holds is easy if the forwarding graph is a line or even a bounded-degree tree; the analysis boils down to proving inductively that the number of
jobs forwarded from one time slot to the next satisfies an exponential tail bound. However, once the graph has cycles, inductive arguments no longer apply. A key part of our proof consists of showing that, among bounded
degree graphs, a bounded degree tree will maximize the probability of overload at any time slot. This requires the use of a ``decorrelation lemma" that allows us to upper bound the probability of bad dependent events by the probability of bad independent events.


Unfortunately, though, the story doesn't end here, because the forwarding graph induced by our pricing does not in general yield
 a bounded degree forwarding graph.\footnote{Figure~\ref{fig:HighIndegree} shows a concrete example where the in-degree of a time slot can be unbounded.} 
Nonetheless, we show that the paths created by the forwarding of jobs across time slots possess a simple canonical form that allows us to modify them and obtain a new forwarding graph of in-degree at most 3, without reducing the load at any resource. 

\paragraph{Beyond unit-length jobs.} We extend the above result to the setting where
each job $j$ requires the use of the resource for
some number of consecutive time units within its window. This is a significantly more complicated problem and, correspondingly, requires a stronger the large market assumption.\footnote{We obtain bounds that match the unit-length case, except for an additional polynomial dependency of the supply requirement on the maximum length of a job. }

\paragraph{Mechanism design for temporal resources.}  As an application of our main theorem, we develop a new online mechanism for selling cloud services when jobs are strategic, that achieves a number of desirable properties in addition to being near-optimal for social welfare. The problem of designing truthful mechanisms for scheduling jobs with deadlines has been studied with many variations in the worst case setting: the parameters the job can lie about (arrival, departure, value, etc.), deterministic vs. randomized, whether payments are determined immediately (prompt) or not (tardy), unit length vs. arbitrary (bounded) length jobs, and assuming certain \emph{slackness} in the deadlines \cite{azar2011prompt,azar2015truthful,hajiaghayi2005online,cole2008prompt,feldman2015combinatorial}. In the worst case setting, the underlying algorithmic problem (that is, without incentive constraints) already exhibits polylogarithmic hardness \cite{canetti1998bounding}.\footnote{The algorithmic problem of stochastic online matching and its generalizations, under large budgets/capacities, are similar in spirit to the stochastic process we consider \cite{GoelMehta,DH09,Devanur2011,agrawal2014dynamic,Feldman10,KTRV14,agrawal2015fast}.  The temporal aspects of the two problems are very different, however, due to which standard models in that literature such as the random order model are not a good fit here.} Lavi and Nisan \cite{lavi2015online} showed that no deterministic truthful mechanism (w.r.t. all the parameters) can approximate social welfare better than a factor $T$ in the worst case, where $T$ is the time horizon, even for unit length jobs on a single machine. Subsequent papers cope with this impossibility by weakening different assumptions. 
In contrast, we consider the Bayesian setting, where jobs are drawn from a known distribution. We give a simple {\em order oblivious posted pricing} mechanism (OPM), where the seller announces prices, and jobs purchase resources in a greedy first-come-first-served fashion. Our mechanism is truthful for jobs' values, requirements, and deadlines; is prompt in that jobs' allocations and payments are determined right at the time of their arrival; and in the stochastic setting, under the large market assumption described above, achieves near-optimal efficiency (a $1-o(1)$ approximation). Determining the pricing requires the seller to know the demand distribution. When the demand distributions are cyclic, say with a period of a week or a month or a year, the optimal prices are also cyclic with the same periodicity. The seller can then use a polynomial size linear program to solve for the appropriate prices. If the demand distribution stays constant over time, then a constant price per unit of resource per unit of time suffices to provide near-optimal system efficiency.

OPMs have previously been shown to achieve constant-factor approximations to 
revenue and social welfare in many different settings. See, e.g., \cite{CHMS-10, 
feldman2015combinatorial}, and references therein. Feldman et al.~\cite{feldman2015combinatorial} show, in particular, that for 
settings with many items and many clients with fractionally subadditive values, 
there always exists an item pricing such that if clients purchase their favorite 
bundles of items sequentially in arbitrary order, the expected social welfare 
achieved is at least half of the optimum. For our setting with temporal 
resources, this implies that when all jobs have unit length, there exists a 
time-of-use pricing that obtains a half approximation to the optimal social 
welfare. In contrast, we obtain a $(1-\epsilon)$ approximation via the same 
kind of selling mechanism under a  large market assumption. Furthermore, while 
Feldman et al.'s approach can only guarantee an $O(\ell_{\max})$ approximation 
when jobs have lengths in $\{1, \cdots, \ell_{\max}\}$,
we are able to use item pricing\footnote{Indeed, Feldman et al. show that in the small market setting, bundle pricing is necessary to achieve an $o(\lmax)$ approximation.} to again obtain a 
$(1-\epsilon)$ approximation under a large market assumption. To our knowledge, 
this is the first near-optimal result $(1-o(1))$ achieved via OPMs that has no dependence on the length of time the system is running (or, in the setting considered by \cite{feldman2015combinatorial},
the number of items being sold).

\paragraph{Other applications.} While the main motivation for our work is to analyze pricing schemes for 
temporal resources, our analysis applies to other resource allocation settings 
where clients have varied preferences over different resources and greedily 
grab the first available resource on their preference list at their time of 
arrival. Consider, for example, a network of charging stations for electric 
vehicles. 
A client wishing to charge his EV strategically chooses which station to obtain service at, depending on the price, travel time, etc.; if that station is already at capacity, the client goes to his next favorite station, and so on. Depending on the geography of the area and traffic patterns, if the forwarding graph over charging stations formed by such a movement of clients has constant in-degree at every node, then our results apply, and near optimal efficiency can be achieved. While pricing had previously been studied in the context of EV charging (see, e.g., \cite{EVcharging-aamas16}), these works focus on optimizing the average case behavior of the system, rather than studying its stochastic behavior.

\paragraph{Connections to queueing theory.}
Special cases of our models are closely related to standard models in queuing theory, where the demand and supply are stationary (i.e., not changing with time). 
In particular, for unit length jobs, suppose that $B_t=B$ for all $t$, the advertised prices are all equal, and every client tries to obtain service at the first slot in its window, failing which it moves its demand to the next time slot, and then the next, and so on.
This case corresponds to the standard  M/D/B queueing model, with Markovian arrivals, deterministic processing time, and $B$ servers, under the first-come first-served (FCFS) queuing discipline.\footnote{The notation for different queuing models is as follows: an A/B/C queue is one where the inter job arrival times are drawn from distributions in family A, the job lengths distributions belong to family B, and there are C identical machines. D is the class of deterministic distributions, M is the class of exponential distributions, GI is the class of general, independent distributions, GIB is the same class with a bounded support, and an $H_2^*$ distribution is a mixture of an exponential and a point mass.} 
Our result matches the {\em optimal bound}  for this model  in the so-called {\em Halfin-Whitt regime} \citep{halfin1981heavy}: if the expected demand in every time period is $ B - {\Omega}( \sqrt{B \log(1/\delta)} )$, then the overload probability is at most $\delta$.  In other words, every arriving job obtains service at the first time slot in its window with probability $1-\delta$.

In the more general model with different job lengths,  the special case of stationary demand and supply  corresponds to the more general M/GIB/B queueing model, where jobs have arbitrary but bounded processing times. 
Even though  optimal bounds in the Halfin-Whitt regime for FCFS have been known for M/M/B and M/D/B queues \citep{erlang1948rational}, GI/M/B queues \citep{halfin1981heavy}, GI/D/B queues \citep{jelenkovic2004heavy}, and GI/$H_2^*$/B queues \citep{whitt2005heavy}, proving such bounds for M/GIB/B queues with FCFS queuing is a major open problem in queuing theory.	
We consider a variant of FCFS: we admit only a certain limited number of jobs of any particular length at every time slot.
Our result in this setting matches the optimal bound in the Halfin-Whitt regime for this variant, 
albeit with a polynomial dependency on the maximum length. 
Our techniques might give a way to prove the same bound for FCFS, which would resolve the open problem regarding M/GIB/B queues mentioned above. 

\paragraph{Organization of the paper.} Our analysis is divided into four main parts. In Appendix~\ref{sec:LP} we describe properties of a time-of-use pricing that balances supply and demand in every time period in expectation; these are summarized in Lemma~\ref{lem:fractional-assignment} in Section~\ref{sec:prelims}. In Section~\ref{sec:network} we analyze the stochastic resource allocation process for a general forwarding graph over resources, and show that the overload probability is related to the in-degree of the forwarding graph. In Section~\ref{sec:temporal} we analyze the temporal resource setting for unit-length jobs, and give a reduction from this setting to a low-degree-forwarding-graph setting. In Section~\ref{sec:non-unit} we extend this analysis to jobs of arbitrary length. A detailed discussion of related work appears in Appendix~\ref{sec:related-full}.

\section{Preliminaries and Main Results}
\label{sec:prelims}

\paragraph{Temporal\comment{The cloud} resource allocation problem.} We consider a setting where a seller has multiple copies of a reusable resource available to allocate over time. 
Clients, a.k.a. {\em jobs}, reserve a unit of the resource for some length of 
time, after which that unit again becomes available to be allocated to other 
jobs. 
A job $j$ is described by a tuple consisting of a starting time, a deadline, a length, and a value, denoted by $(s_j,d_j,l_j, v_j)$, with the first three elements in $\integers_+$ and the last in $\reals_+$.  
The interpretation is that the job can be processed in the time interval $[s_j,d_j]$, and requires $l_j$ consecutive units of time to complete. 
The value accrued by processing this job is $v_j$. 
Let $W_j=[s_j, d_j-l_j+1]$ denote the ``job's window'' or the interval of time during which the job can be started so as to finish before its deadline. 
For each $t \in \integers_+$, at most $B_t \in \integers_+$ jobs can be processed in parallel.

We consider the following stochastic model of job arrival: there is a set of \emph{potential} jobs $J$; associated with each job $j\in J$ is  a probability $q_j$. 
A potential job $j$ is realized with probability $q_j$ via an independent Bernoulli trial. The order of arrivals of the realized jobs in the system is determined by an adversary who knows the set of realized jobs.
\footnote{ For example, a job $j$ that shows up well before $s_j$ may make a reservation for resources in its window in advance.}


A scheduling mechanism, at the time of each job's arrival, determines whether or not to accept a job. In the former case it allocates $l_j$ consecutive units of time in the time interval $[s_j,d_j]$, and charges the job a payment $p_j$. 
Job $j$ derives a \emph{utility} of $v_j - p_j$ if it is accepted, and 0 otherwise.\footnote{We assume that the utility of job $j$ is $-p_j$ if it does not get at least $l_j$ units of time within the interval $[s_j,d_j]$.}  The objective of the algorithm is to maximize the total value of the jobs processed, a.k.a. the \emph{social welfare}.  The mechanism is required to be \emph{truthful} in dominant strategies, which means that a job $j$ can not get a higher utility by misreporting any of its parameters.\footnote{We assume that the setting disallows Sybil attacks. That is, a job of length $l$ cannot pretend to be multiple different ``subjobs'' of total length $l$ trying to obtain service in consecutive time blocks.}  The algorithm knows the set of potential jobs $J$ (each defined by its associated 4-tuple as above), their arrival probabilities, and the capacities $B_t$ ahead of time, but the realized job arrivals are learned as they happen. The scheduling decision and payment are determined at the time of the jobs' arrival and are irrevocable.

\paragraph{Time of Use pricing.} 
We consider a particularly simple kind of mechanism that announces a ``time of use pricing''  $(p_t)_{t\in\integers_+}$ up front, where
$p_t$ is the price per unit of resource at time $t$. The mechanism then requires a job of length $l_j$ starting at time $t$ to pay a total price of $p_t(l_j)=\sum_{t'=t}^{t+l_j-1} p_{t'}$. 
For every job $j$, let $\favj = \arg \min_{t \in  W_j: p_t(l_j)\le v_j} \{ p_t(l_j)\}$ denote the job's least expensive options within its window, a.k.a. its ``favorite'' starting slots. A mechanism that assigns every arriving job to one of its favorite slots is trivially truthful.
It follows from strong LP duality by a standard argument that with an appropriate choice of prices, such a mechanism obtains nearly the optimal social welfare, \emph{if it is only required to satisfy the supply constraints in expectation}. See Appendix~\ref{sec:LP} for the LP and a proof.
 Let $\opt$ denote the expected maximum social welfare achievable by any feasible (capacity respecting) assignment under this stochastic arrival model.

\begin{lemma}
\label{lem:fractional-assignment}
{\bf (Fractional assignment lemma)}
Fix any set of potential jobs $J$, their arrival probabilities, and the capacities $B_t$ for all $t\in \integers_+$. Then for any $\epsilon > 0$,  $\exists$ nonnegative prices 
$(p_t)_{t\in\integers_+}$ and a fractional assignment $X_{j,t}\in [0,1]$ from 
jobs $j\in J$ to their favorite slots $t\in \favj$,
such that,
\begin{enumerate}
  \item Every job that can afford to pay the price at its favorite slot is 
  fully scheduled: for every $j$ with $p_t(l_j)<v_j$ for $t\in\favj$, we have 
  $\sum_{t\in \favj} X_{j,t}=1$.
	\item The expected allocation at time $t$ is at most $(1-\epsilon)B_t$: 
	$\forall t, \sum_{j\in J, t' \in [t-l_j+1,t]  } q_j X_{j,t'} \leq 
	(1-\epsilon)B_t$. 
	\item The expected social welfare is at least $(1-\epsilon)$ times the 
	optimum: $\sum_{j \in J, t\in \favj } v_j q_j X_{j,t} \geq 
	(1-\epsilon)\opt$. 
\end{enumerate} 
Further, if the distribution is periodic,\footnote{See formal definition in Appendix \ref{sec:LP}.} the prices are also periodic with the same period, and can be computed efficiently.
\end{lemma}

\paragraph{The asynchronous allocation process.}
Of course, the actual allocation of slots to jobs happens in an online fashion and the capacity constraints are hard constraints that must
be met regardless of
which jobs are actually realized.
The mechanism we analyze is a greedy first-come first-served\footnote{Sometimes for jobs of length $>1$ our mechanism artificially limits the capacity of a slot, that is, does not allocate a block of time slots even when available; however, it does so in a truthful manner. This detail is discussed in Section~\ref{sec:non-unit}.} type mechanism: The slot prices $(p_t)_{t\in\integers_+}$ induce a preference ordering over
slots for each job $j$;
this is a list of slots $t$ in $j$'s window $W_j$ with $p_t(l_j)<v_j$, in non-decreasing order of price. 
Let $\Pathj$ denote the preference ordering\footnote{When prices of time slots are not unique, this preference ordering is not unique. We need to impose a particular tie breaking rule among the job's favorite slots, but can break ties among other slots arbitrarily. Part of this tie-breaking is required already to satisfy the conclusion of \Cref{lem:fractional-assignment}, part of it is required to ensure the overall stability of this system.  We detail the tie-breaking rule in \Cref{sec:temporal}. Note that the mechanism remains truthful regardless of the tie-breaking rule.} of job $j$ over time slots in its window.
When job $j$ arrives, it considers  time slots in the order of $\Pathj$, and gets served at the first one that has resources available (or doesn't get served
if no slot in $\Pathj$ has leftover capacity). 
We emphasize that {\em which} jobs are realized 
is determined by the stochastic model described above, but {\em when} jobs arrive is determined adversarially, and can depend
 on which other jobs are realized.
\footnote{This is a reversal of the random order model popular in online matching, where the set of arrivals is adversarial but the order is random.}  For this reason, we call this an ``asynchronous allocation process''. 
 


Our main theorem shows that with an appropriate choice of $\epsilon>0$, 
the asynchronous allocation process corresponding to the price vector given by Lemma \ref{lem:fractional-assignment} obtains near-optimal social welfare. 



Let $\lmax:= \max_{j \in J} l_j$ and $B := \min_t B_t$. The case when $\lmax = 1$ is called the \emph{unit length jobs} setting.
\begin{thm} 
\label{thm:stability-cloud}
{\bf (Stability of service theorem)}
  $\exists$  a universal constant $c$ such that $\forall~\epsilon\in [0,1/2]$, for prices determined by Lemma~\ref{lem:fractional-assignment} for this $\epsilon$, in the asynchronous allocation process for the temporal\comment{cloud} resource allocation problem, every arriving job that can afford the price at its favorite slot gets accepted at such a slot with probability $\geq 1-\epsilon$, and the social welfare achieved is $\geq (1-2\epsilon)$ times $\opt$, if  for the unit length jobs case and the general case respectively,
\[  B  \ge c \frac {\log(1/\epsilon)} {\epsilon^2}, \quad  \text{ and, } \quad
B \ge c\frac{\lmax^6\log(1/\epsilon)}{\epsilon^3} .\]
\end{thm}


As a step towards proving this theorem, we will study a slightly more abstract setting without prices in the next section: Suppose that the time slots are nodes in a ``forwarding graph'' $G$, and that there
is an edge from time slot $t'$ to time slot $t$ if there is some job $j$
such that $t$ follows $t'$ in $j$'s preference ordering $\Pathj$.
Jobs arrive at the various nodes in the graph\footnote{ 
The prices from Lemma \ref{lem:fractional-assignment} will guarantee that the number
of external arrivals to each node is slightly less than the capacity
of that node.} and
 move through this graph until they are successfully served.
However,  we relax the requirement that each
job $j$ must follow $\Pathj$; rather, we allow
each job to take an adversarially selected path in the forwarding graph $G$ in its quest for service.
We then present conditions on the arrival process, in terms
of the maximum indegree of the graph, under which failures are 
unlikely to cascade.

\section{Stability of service for a network of servers}
\label{sec:network}

We will analyze the temporal process described above by
reducing it to the following {\em network of servers} setting: 
\begin{itemize}
\item There is  a set of $n$ {\em servers}, which we identify with $[n]$. Server $i$ can service a total of $B_i$ jobs and then expires.

\item There is a directed forwarding graph $G$ whose vertex set is the $n$
servers $[n]$.  Let $\dmax$
denote the maximum indegree in $G$. We will refer to
the vertices of $G$ as either servers or nodes. 

\item  The number of jobs entering the network at each node is 
determined by a stochastic
process\footnote{ E.g., Potential job $j$ 
arrives with probability $q_j$. } : Denote
by $A_i$ the number of jobs that enter the network at
node $i\in [n]$. The random variables $A_i$ are mutually independent.

\item  Each arriving job $j$ is forwarded through the network $G$
until the job reaches an {\em available} server. Server $i$ is available if it has not yet served $B_i$ jobs.
Thus,  if job $j$ enters the network at $i$, and server $i$ is available,
$i$ serves $j$ and $j$ leaves the network. If $i$ has already served $B_i$ jobs prior to $j$'s arrival,
then job $j$ gets {\em ``forwarded''} to some neighbor of $i$ in $G$ and tries
to get service there, and so on. Job $j$ leaves the network as soon
as it is serviced, or it has tried all reachable servers, or it gives up, whichever happens first.

\item All aspects of this process other than the external arrival process are
assumed to be {\em adversarial}: the paths jobs take as they seek an available server, the timing of external arrivals, and the timing of forwarding events. 
\end{itemize}

Our main theorem for this setting gives conditions under which
the probability that any particular job gets served by the first server it tries is close to $1$. The crucial point here is that this is independent of $n$.

\begin{thm}
\label{thm:Network-main}
Consider the network of servers setting as above. 
Fix an $\epsilon$ in $[0,1/2]$. Suppose that for each node $i\in[n]$
the moment generating function of $A_i-B_i$  satisfies 
$\EE { e^ {\epsilon (A_i- B_i)}} \le \epsilon ^2/e\dmax.$
Then for any job $j$, $\Pr[{j \text{ is not served at the first
    node on its path}}] \le \epsilon$.
In other words, failures don't cascade and each job is served
with high probability at the node at which it enters.
\end{thm}


We introduce some more notation before we proceed. 
 For a particular instantiation of the 
 process (as determined
 by the stochastic job arrivals and the adversarial timing of arrivals
 and forwards), let $P_j$ denote an arriving
 job $j$'s {\em realized path} in $G$, the set of servers that $j$ tries to get
 service from. This path begins of course at the node at which
 $j$ enters the network. Let
 $\realPaths=(P_j)$ denote the collection of all realized paths.
Let $\ell_i(\realPaths)$ denote the number of jobs that attempt to get
service at  node $i$ (external job arrivals to $i$ as well as
forwards), a.k.a. its ``load''. We say that node $i$ is ``overloaded'' if
$\ell_i(\realPaths)\ge B_i$.

If a node $i$ forwards
a  job, then 
node $i$ must have already served $B_i$ other jobs. Thus, the collection of
realized paths $\realPaths$ satisfies the following {\bf min-work
	condition}: for every node $i$, the number of jobs forwarded is
no more than the number of realized paths $P_j$ containing $i$ minus the
capacity $B_i$.

We now proceed to sketch a proof of Theorem \ref{thm:Network-main}. A detailed proof can
be found in Appendix~\ref{sec:network-full}.

%
Consider the load on a single node and suppose that it has constant in-degree. If each of the forwards from its predecessors 
were {\em independent}, and these forwards were few and far between, 
as captured by a bound on the expectation of the moment generating function, 
then it can be argued that forwards from this node would also inductively satisfy 
a similar bound on its moment generating function. 
The forwards are not independent, so this simple approach does not work. 
Moreover, $G$ is not necessarily acyclic, so there is not even an obvious order for induction. 
However, these conditions are satisfied when $G$ is a tree, 
and our first lemma formalizes the above approach in this case. 

\begin{lemma}
\label{lemma:TreeProcess}
Fix an $\epsilon$ in $[0,1/2]$. Suppose that the network $G$ is a
finite directed tree, that is, it contains no directed cycles and every node has out-degree $1$, and that the moment generating function of $A_i-B_i$
for each node $i$ satisfies $\EE { e^ {\epsilon (A_i- B_i)}} \le
\epsilon ^2/e\dmax$.
Then, for any $i$,
$\Pr[{\ell_i(\realPaths)\ge B_i}] \le \epsilon$.
\end{lemma}


\afterpage{%
\thispagestyle{empty}
\begin{figure}[H]
\begin{center}
\includegraphics[width=\textwidth]{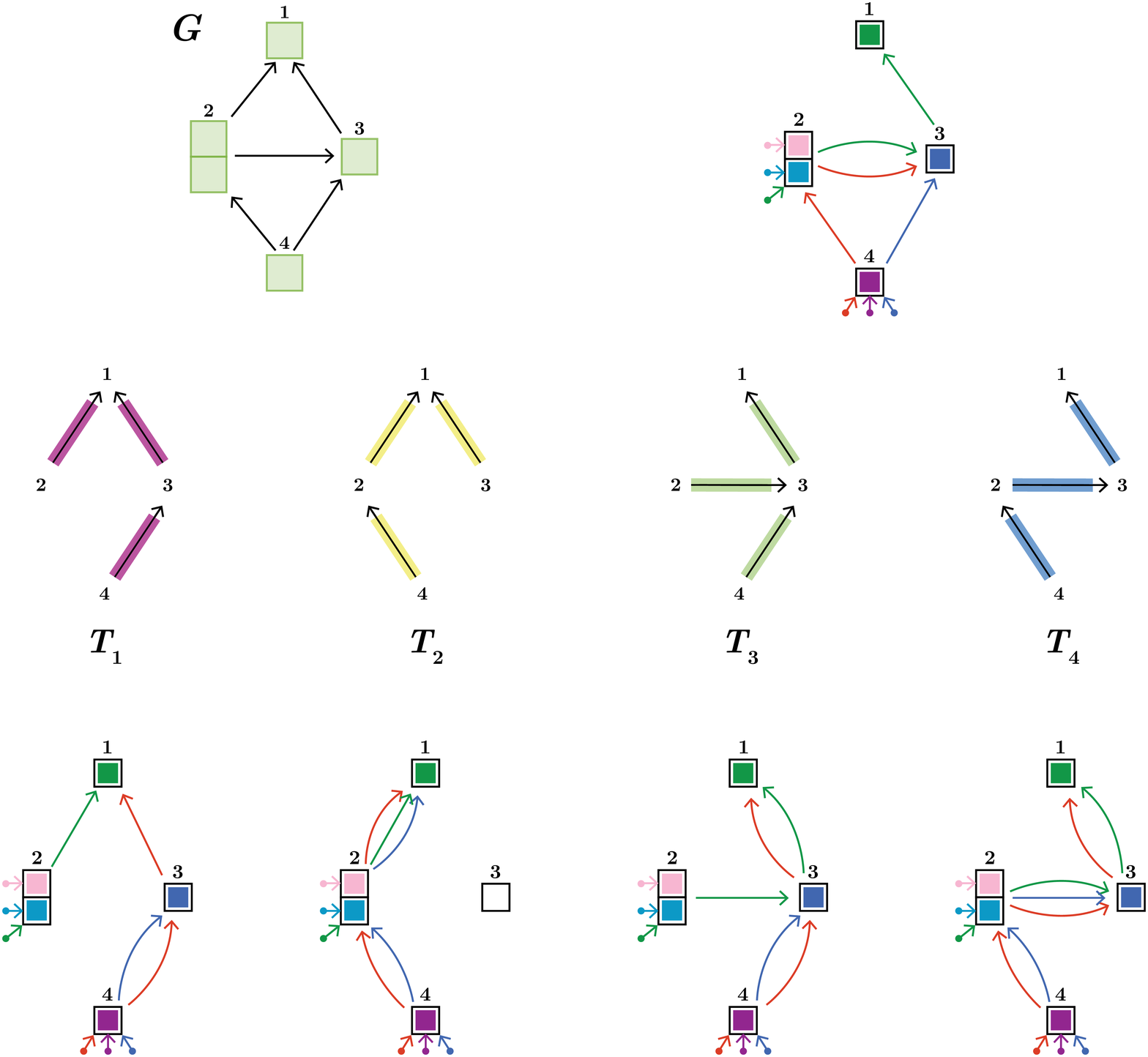}
\end{center}
\caption{
\label{fig:Routing}
This figure illustrates the first step of the argument for the network
of servers. The graph $G$ is shown on the upper left. In this example, all nodes except for node 2 have a capacity of 1; node 2 has a capacity of 2.
The upper right shows an example of what might happen with three jobs arriving at node 2 and three arriving at node 4.  The blue job arriving at node 4 gets forwarded to node 3 where it gets
served. The green job arriving at node 2 gets forwarded to node 3 and then to node 1 where it finally gets service, and so on.  Notice that the load at node 1 for this set of job arrivals and paths is 1.  The middle panel of the figure shows
the four trees in  $\Trees{1}$. The bottom panel shows how the packets might be routed and served, if all forwarding was done along edges of the associated tree (immediately above). In this example, the worst case load at node 1 is when all jobs
are routed along edges of the tree $T_2$.  This results in a load of 3 at node 1.}
\end{figure}

\newpage

\thispagestyle{empty}

\begin{figure}[H]
\begin{center}
\includegraphics[width=0.8\textwidth]{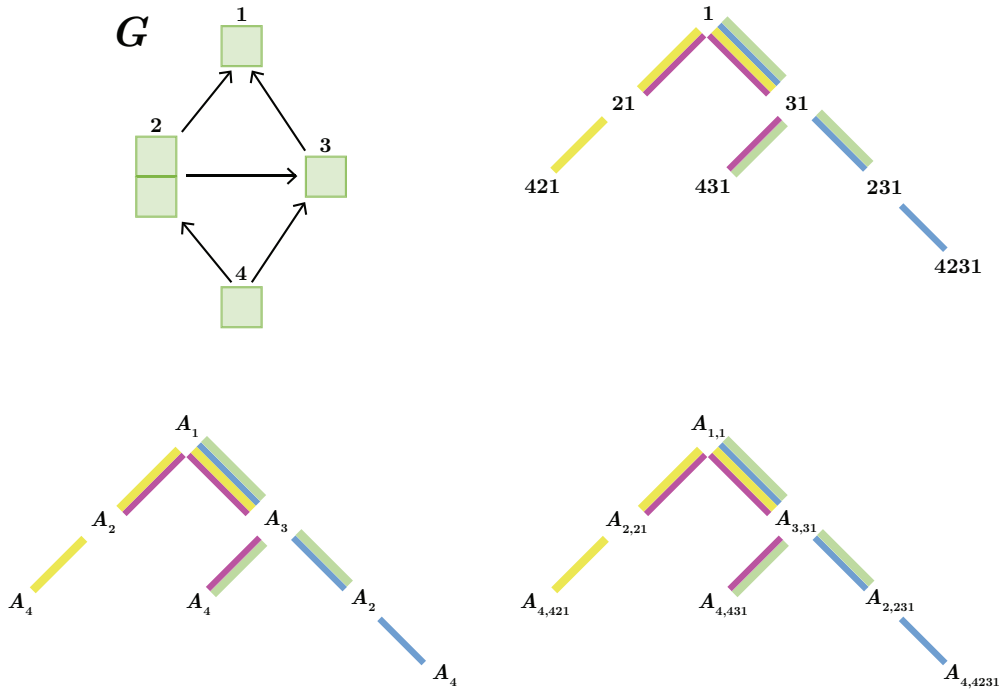}
\end{center}
\caption{
\label{fig:AllMulticolorTrees}
\small{This figure illustrates the second step of argument for network of servers.
Again on the upper left we see the graph $G$. On the upper right
we see the tree of trees $\binaryTree _1$. Each node is labeled with a simple
path in $G$ to node 1. It is easy to see that each of the trees $\Trees{1}=\{T_1, \ldots, T_4\}$ 
from Figure \ref{fig:Routing} is a subtree of   $\binaryTree _1$. (The edges
are color-coded as in the corresponding tree in $\Trees{1}$.)  The bottom left version of $\binaryTree _1$ indicates on each node the random variable which is
the number of external job arrivals at that node. Lemma \ref{lem:step1-Tree} implies that
it suffices to bound the probability of overflow at the root of this tree.
Lemma  \ref{SDApp} shows that instead we can bound the probability
of overflow at the root of the bottom right tree, where the external arrivals at different nodes are independent. For example, $A_{2,21}$ and $A_{2, 231}$ are independent samples from the distribution of $A_2$.} }
\end{figure}

\begin{figure}[H]
\begin{center}
\includegraphics[width=0.23\textwidth]{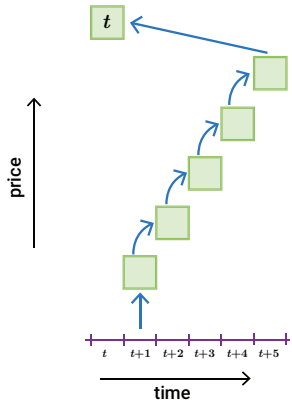}
\end{center}
\caption{
\label{fig:HighIndegree}
\small{This figure illustrates that the temporal process can have high in-degree. 
Suppose that there are jobs with window $[t, t+1]$, $[t, t+2]$,
$[t, t+3]$ and so on. Then all of these jobs would first try slot $t+1$. 
The final slot all of these jobs would try is slot $t$, so there would be
an edge in the forwarding graph from each of $t+1, \ldots, t+5$ to $t$.}}
\end{figure}
}


Our proof of the Theorem \ref{thm:Network-main} will reduce the analysis in a general
network  to that in an appropriately defined tree network. The argument has two parts that we outline next.
Throughout the proof, we will focus on a particular server $u$ in $G$.

\subsection*{Part 1: Reducing to a tree for fixed arrivals}
In the first part, we fix the set of realized paths $\realPaths$ (as determined
by the stochastic job arrivals and the adversarial timing of arrivals
and forwards).
This fixes the entries $\ba = (a_i)_{i=1}^n$, where $a_i$ is
the number of jobs arriving at node $i$ from outside the network.
We then show that if the node $u$ is overloaded for this fixed outcome,  
then there exists a subtree of the network $G$ that is rooted at $u$, such  that 
if jobs are forwarded exclusively along edges of this tree until service is received (or there is no where else to go), then node $u$
is still overloaded. 

\addtocounter{page}{-2}

More formally, let $T$ be a directed tree rooted at the node $u$.  For a vector of external arrivals $\ba =
(a_i)$ and node $i$, let $\ell_i^T(\ba)$ denote the load on node $i$
(external arrivals plus forwards) when jobs are forwarded along the edges of the tree $T$ until service is received.

Let $\Trees{u}$ denote the set of all directed subtrees of
$G$ rooted at node $u$.  The following lemma captures the first part
of our analysis.

\begin{lemma}
\label{lem:step1-Tree}
If a fixed set of arrivals (resulting in a particular $\ba$) and induced paths $\Paths$ overload a
node $u$ in the network $G$, then  $\exists$ a tree $T\in\Trees{u}$
such that $u$ is overloaded with the same set of arrivals $\ba$ when
requests are routed along $T$. Formally, 
$$\textstyle  \ell_u (\Paths) \ge B_u \quad\text{ implies that }\quad \max_{T\in
  \Trees{u}} \ell_u^T (\ba) \ge B_u.$$
\end{lemma}
The first step in the proof consists of removing cycles in $\realPaths$ while preserving the set of overloaded vertices.  
In the second step we reroute the paths so that they form a tree. 
The proof of this lemma is deferred to the appendix. 
An example is presented in Figure~\ref{fig:Routing}.

\subsection*{Part 2: Reducing to a tree of trees}

Lemma \ref{lem:step1-Tree} does not reduce the analysis of the network of servers setting  to the analysis of a single tree, because the particular tree
that gives the worst-case load on $u$
depends on the realized arrivals $\ba = (a_i)_{i=1}^n$.
The lemma does show, however, 
that to complete the proof of Theorem \ref{thm:Network-main} it
suffices for us to bound the probability $\Pr\left[{\max_{T
      \in \Trees{u}}\ell_u^T (\bA) \ge B_u}\right]$ for each node $u$,
      where $\bA :=(A_i)_{i=1}^n$. (Recall that random variable $A_i$ denotes the number of external arrivals at node $i$, 
and that  the different $A_i$'s are mutually independent.)

In order to analyze this quantity, we will construct a new tree
network\footnote{ We call this the {\em tree of trees}.} $\binaryTree _u$ over an expanded set of nodes that contains
{\em every} tree $T\in \Trees{u}$ as a subtree. The tree $\binaryTree _u$ is defined as follows: There is a node $v_P$ in $\binaryTree _u$ for each simple directed path $P$
in $G$ terminating at $u$, and there is an edge in $\binaryTree _u$ from $v_P$ to $v_{P'}$, if   $P=iP'$ for some node
$i$ in $G$. By construction, each tree $T$ in $\Trees{u}$ has a unique isomorphic copy
in $\binaryTree _u$. See Figure \ref{fig:AllMulticolorTrees}. \footnote{ This construction blows up the number of nodes exponentially, but this does not affect us since 
our bound (\Cref{lemma:TreeProcess}) is independent of the number of nodes in the tree.}

We then consider the network of servers process on $\binaryTree_u$,
under the assumption that for every node $v_P$ such
that $P=iP'$, the number of external arrivals at $v_P$
is $A_i$, and also that as long as a job is not serviced it is forwarded
along the next edge in the tree.
 Then for any tree
$T\in\Trees{u}$, the load on the node corresponding to $i$ in the isomorphic copy of $T$ in
$\binaryTree_u$ is no smaller than the load on $i$ 
in $T$ under the same set of arrivals. In particular,
\begin{equation}\label{eq:treeoftrees}
\textstyle \Pr\left[{\ell_u^{\binaryTree _u} (\bA) \ge B_u}\right]\ge \Pr\left[{\max_{T \in \Trees{u}}\ell_u^T (\bA) \ge B_u}\right].
\end{equation}

Unfortunately, we cannot analyze $\ell_u^{\binaryTree _u}(\bA)$ as in
the proof of Lemma~\ref{lemma:TreeProcess}, since the external arrivals at
different nodes are correlated.  In particular, for each node $i$ in $G$,
there are $n_i$ nodes in $\binaryTree _u$ at which the entries
$A_i$ are the same, where $n_i$ is the number of different directed
simple paths from $i$ to $u$ in $G$.  The key step in the rest of the
proof is to show that replacing these by independent draws from the
same distribution can only (stochastically) increase the load at $u$. 
To this end, we require the following ``decorrelation'' lemma:
\begin{lemma}
	\label{lem:probmore}
{\bf (Decorrelation lemma for the max function)}
	Let $g_{\ell}:\Re \rightarrow \Re$, $\ell=1, \ldots, k$ be any non-decreasing functions,  $X$ be any real valued random variable, and let $Y_1, \ldots, Y_k$ be independent and identically distributed random variables from the same distribution as $X$.
	Then, 
	\[ \textstyle \max_{\ell}\{g_{\ell} (Y_{\ell})\} \stgeq \max_{\ell}\{ g_{\ell}(X)\}  ,\]
	where
$\stgeq$ denotes stochastic dominance.
\end{lemma}

Applied to our setting, the decorrelation lemma gives us the following result.
\begin{lemma}
\label{SDApp}
For each $i$ and directed simple path $P$ from $i$ to $u$, let
$A_{i,P}$ be an independent draw from the distribution of $A_i$, let 
$P_i(T)$ be the unique path from $i$ to $u$ in tree $T$, and let
$\stgeq$ denote stochastic dominance. Then
$$\textstyle \max_{T \in \Trees{u}}  \ell_u^T \left((A_{i, P_i(T)})_{i \in T}\right)\stgeq\max_{T \in \Trees{u}}\ell_u^T (\bA).$$

\end{lemma}
Theorem \ref{thm:Network-main} now follows by observing that (\ref{eq:treeoftrees}) holds also w.r.t. the arrivals $(A_{i, P})_{i \in [n], P\in \binaryTree_u}$. 
Then  $ \ell_u^{\binaryTree_u} \left(  (A_{i, P})_{i \in [n], P\in \binaryTree_u} \right)$ can be analyzed using \Cref{lemma:TreeProcess} since 
we now have the required independence. 
Lemmas~\ref{lem:step1-Tree} and \ref{SDApp} complete the sequence of inequalities.

\section{Stability of service for unit-length jobs}
\label{sec:temporal}

We now return to the temporal resource allocation problem
and prove the stability of service theorem,
Theorem~\ref{thm:stability-cloud}, for the special case where each job
has unit length, that is, $l_j=1$ for all $j$. The non-unit length case is
discussed in Section~\ref{sec:non-unit}.

We fix $\epsilon$ as stated in the theorem, as well as the set of
prices given by the fractional assignment lemma
(Lemma~\ref{lem:fractional-assignment}). Then, for the asynchronous
assignment process induced by these prices, we construct an instance
of the network of servers setting discussed in Section~\ref{sec:network}
that satisfies the assumptions made in
Theorem~\ref{thm:Network-main}. Applying that theorem would then imply
Theorem~\ref{thm:stability-cloud}.

The obvious way to reduce from the temporal setting to
the network of servers setting was described at the end of Section \ref{sec:prelims}:
construct a forwarding graph $G$ over the set of all time
slots $t\in \integers_+$ so that it
contains all edges $(t, t')$ that are in some
job's preference order  $\Pathj$ over time slots.\footnote{ In other words,
$t'$ is the next slot after $t$ that some job $j$ prefers, given that job's
window and the slot prices.}
Unfortunately, the graph so defined can have
unbounded in-degree. See Figure~\ref{fig:HighIndegree}.  Observe  though 
that in
this example, the path of every job that is forwarded to node $t$ goes through the node $t+1$. As such, each of these jobs
is effectively forwarded from $t+1$ to $t$. Taking
inspiration from this example, we will proceed as follows. For every
instantiation of job arrivals and preference orderings, we will define
a canonical ``shortcutting'' of the jobs' paths, such that the
overload status of every time slot is maintained. We will then show
that the union of the shortcut paths over all possible instantiations
defines a bounded degree graph. We can then apply
Theorem~\ref{thm:Network-main}.

We give a brief overview of this argument below. Details can be found
in Appendix~\ref{sec:temporal-full}. 

\paragraph{The network of servers.}
We begin by defining a directed graph $D$ on the set of all time slots
$\integers_+$ as follows. For every time slot $t\in \integers_+$, define 
$\ell(t) =
\max \{s<t: p_s\le p_t\}$ and $r(t) = \min \{s>t: p_s< p_t\}$ to be
the left and right ``parents'' of $t$. Let $E_F := \{ (\ell(t), t) \cup
(r(t), t ) \quad \forall t\in \integers_+\}$; we call this the set of {\em
  forward} edges. Let $E_B: = \{ (b(t), t) \quad \forall t\in \integers_+\}$
where $b(t)= \min \{s>t: p_s= p_t\}$; we call this the set of {\em
  backward} edges. The directed graph $D$ on vertex set $\integers_+$ is then 
  defined as 
  $D := (\integers_+,E_F \cup 
  E_B)$. Observe
that every node $t\in \integers_+$ in this graph has in-degree at most $3$.
Figure \ref{fig:BoxesRedBlue} illustrates the forward edges in this construction.

Let $\tilde D$ denote the graph formed by just the forward edges:
$\tilde D = (\integers_+,E_F)$. For any $t\in \integers_+$, let $\Anc(t)$ 
denote the {\em
  ancestors} of $t$ in $\tilde D$, that is, $\Anc(t) = \{s \text{ such
  that there is a path in }\tilde D\text{ from } s\text{ to }t\}$.

\paragraph{The reduction.}
We now consider the network of servers setting over the graph $D$, and
describe a specific realization of jobs and paths for every
realization of jobs and paths in the temporal setting. The set of
arriving jobs and their entry nodes are the same in the two
settings. We need to redefine the realized paths of the jobs to
follow the edges in $D$. 

Recall that in the temporal setting, each arriving job had a preference
ordering $\Pathj$ over time slots in its window. We complete the description of $\Pathj$ by specifying how ties are broken:
$\Pathj$ begins at the
node $y_j\in\favj$ to which it is assigned in the fractional
assignment\footnote{Note that this node may be a random variable, but
	is always among the favorite nodes of the job. The job is
	indifferent over all the nodes over which we tie-break.}  returned
by Lemma~\ref{lem:fractional-assignment}. It then visits other nodes
in $\favj$, if any, in a particular order: first, it visits all nodes
$t\in \favj$ with $t<y_j$ in decreasing order of time, then it visits
all nodes $t\in \favj$ with $t>y_j$ in increasing order of
time. Having visited all of the least price slots in its window, the job
then visits all slots of the next smallest price in its window in
increasing order of time, and so on. See Figure~\ref{fig:BoxesLRU} for
an illustration.

Let $P_j$ be the realized path of job $j$, namely, the prefix of $\Pathj$ from
$y_j$ to the node (call it $z_j$) where the job receives service or exits 
the process. We use $P_j^1$ to denote the prefix of this
path which visits nodes $t\in \favj$ with $t<y_j$ in decreasing order
of time; this always contains the node $y_j$. The remaining suffix of
$P_j$, if non-empty, is denoted $P_j^2$. Observe that every edge in
$P_j^1$ is a backward edge. However, edges in $P_j^2$ don't
necessarily belong to $D$.

Let ${\tilde P}_j^2 = P_j^2\cap \Anc(z_j)$, in other words, we remove
from $P_j^2$ all of the nodes that are not ancestors of $z_j$ in the
graph formed by the forward edges, $\tilde D$. The resulting path is a
{\em short-cut} of the original path of the job. We now define
${\tilde P}_j$, a path from $y_j$ to $z_j$, as follows. If $z_j\in
P_j^1$, then $\tilde P_j := P_j^1$; otherwise, we define an
appropriate prefix of $P_j^1$ called $\tilde P_j^1$, and set $\tilde
P_j = \tilde P_j^1\cup\tilde{P}_j^2$. Observe that $\tilde P_j$ is a
short-cutting of $P_j$. Furthermore, every node that we short-cut in
this process forwarded the job $j$, and is therefore overloaded. We
can now prove the following two lemmas. See
Figure~\ref{fig:Shortcutting} for an illustration of the short-cutting
procedure.

\begin{lemma}
\label{lem:path-reduction-1}
 Paths ${\tilde P}_j$ as defined above lie in the graph $D$.
\end{lemma}

\begin{lemma}
\label{lem:path-reduction-2}
The collection of paths $\tilde\realPaths=({\tilde P}_j)$, as defined
above, satisfies the min-work condition. Further, 
a node $t\in \integers_+$ is overloaded under the realized paths
$\tilde \realPaths$ if and only if it is overloaded under the realized
paths $\realPaths$. That is, $\ell_t(\realPaths)\ge B_t$ if and only if
$\ell_t(\tilde\realPaths)\ge B_t$.
\end{lemma}


\afterpage{%
\thispagestyle{empty}
\begin{figure}[H]
\begin{center}
\includegraphics[width=0.37\textwidth]{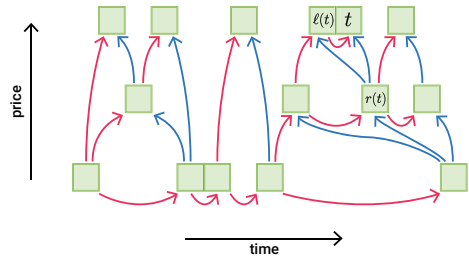}
\end{center}
\caption{
\small{This figure shows the set of all forward edges in the directed, acyclic graph $D$ on a set of time slots. Each time slot is represented by a green square, with its height indicating its price. Red edges go from $\ell(t)$ to $t$, and blue edges go from $r(t)$ to $t$ for each $t$.}}
\label{fig:BoxesRedBlue}
\end{figure}

\begin{figure}[H]
\begin{center}
\includegraphics[width=0.37\textwidth]{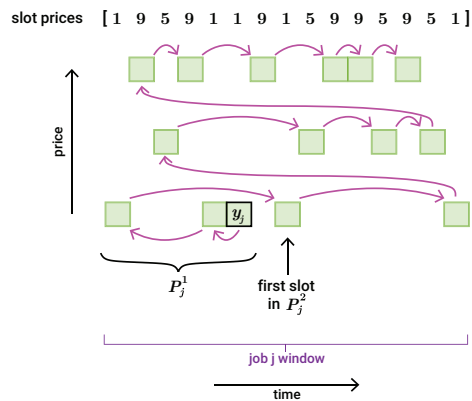}
\end{center}
\caption{
\label{fig:BoxesLRU}
\small{This figure shows the canonical path of a job over the time slots.
The top line shows the prices of each of the time slots. The job enters at slot $y_j$. The decomposition of the path into $P_j^1$ and $P_j^2$ is also illustrated.}}
\end{figure}

\begin{figure}[H]
\begin{center}
\includegraphics[width=0.8\textwidth]{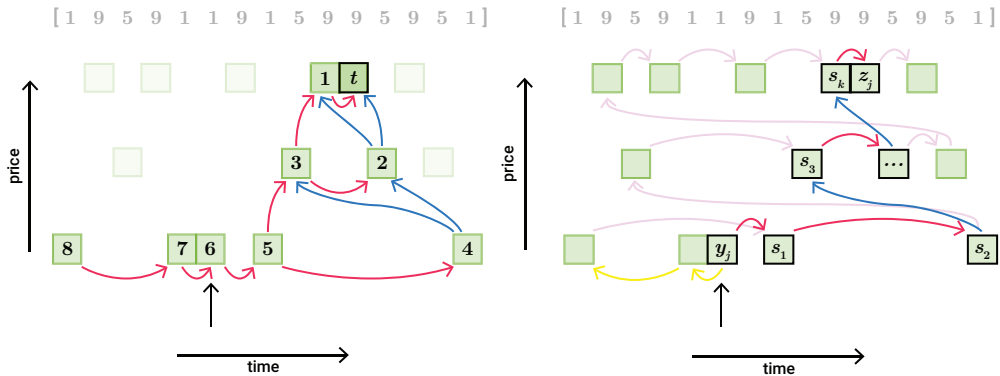}
\end{center}
\caption{
\label{fig:Shortcutting}
\small{The figure on the left displays ancestors of $t$ in $\tilde D$ numbered in reverse topological order. The figure on the right displays
the shortcutting of the path from Figure~\ref{fig:BoxesLRU}. The yellow path is $P_j^1$, and the red and blue path is $\tilde{P}_j^2$. In this case, $\tilde{P}_j = \{y_j\}\cup\tilde{P}_j^2$.}}
\end{figure}
}

We are now ready to prove the stability of service theorem for
unit-length jobs. Observe that the instance of the network of items
setting described above satisfies all of the properties required by
Theorem~\ref{thm:Network-main}. In particular, for every time slot
$t$, the number of arrivals $A_t$ is given by $\sum_j q_j
\hat{X}_{j,t}$, where $\{X_{j,t}\}$ is the fractional assignment given
by Lemma~\ref{lem:fractional-assignment} and $\hat{X}_{j,t}$ is a
Bernoulli random variable with expectation $X_{j,t}$. Therefore, it
can be verified that $\EE { e^ {\epsilon (A_t- B_t)}} \le \epsilon
^2/3e$ for all $t$, and every job gets serviced with probability at
least $(1-\epsilon)$ times its total fractional assignment.

\addtocounter{page}{-1}

\section{Stability of service for arbitrary length jobs}
\label{sec:non-unit}

We now turn to the temporal resource allocation problem for jobs of
arbitrary lengths, and prove Theorem~\ref{thm:stability-cloud}. 
As in Section~\ref{sec:temporal}, we fix any $\epsilon>0$, and a set
of prices for the time slots as given by
Lemma~\ref{lem:fractional-assignment} for this $\epsilon$. Recall that
$p_t(l)$ denotes the total price for $l$ consecutive units of resource
starting at time $t$. A job $j$ of length $l_j$ can choose to buy
$l_j$ or more consecutive units of resource depending on availability
at these prices; we call these consecutive units ``time blocks'' and
denote them by the pair $(t,l)$ where $t$ is the starting time of the
block and $l$ its length. The prices induce for each job $j$ a
preference ordering $\Pathj$ over time blocks $(t,l)$ with $t\in W_j$
and $l\ge l_j$, and ties broken appropriately. As in the unit-length
case, jobs search for the first available time block in their
preference ordering in adversarial
order. Lemma~\ref{lem:fractional-assignment} guarantees that for every
time slot $t$, the expected number of arriving jobs whose first block
in their preference ordering starts at $t$ is at most
$(1-\epsilon)B_t$.

\paragraph{Correlation introduced by non-unit length jobs.} As in
Section~\ref{sec:temporal} we can think of the movement of jobs as
inducing paths in a graph over (starting) time slots. The challenge
with non-unit length jobs is that when considering a block $(t,l)$,
they need to check the availability of the resource at each of $l$
different slots; in other words, the forwarding decision for such jobs
at slot $t$ depends on loads at other neighboring slots, introducing
extra correlations in the forwarding process. Alternately, we can
think of the movement of jobs as inducing paths in a graph over {\em
  time blocks}. The challenge now is that we don't have a well defined
notion of capacity; rather each time block shares capacity with other
overlapping time blocks in a non-trivial manner.

\paragraph{Solution: capacity partitioning.} We adopt the second
approach. In order to overcome the challenge described above, we
decouple capacity constraints at time blocks by artificially limiting
the number of jobs assigned to any block. In particular, we assign a
capacity of $\Btl$ to time block $(t,l)$. Once $\Btl$ jobs have been
assigned to block $(t,l)$, even if there are available resources at
all slots in the interval $[t, t+l-1]$, we admit no more jobs at this
block. In order to respect the original capacity constraint at a time
slot $t\in \integers_+$, the capacities $\Btl$ must satisfy for all
$t$ the property that $\sum_{l\in[l_{\max}]} \sum_{t'\in [t-l+1, t]} \Btl[t']
\le B_t$.

Two issues remain: (1) How should the capacities be set to satisfy
the above per-slot capacity constraints while obtaining good social
welfare? (2) What process/graph does this induce over time blocks?

\paragraph{Setting the capacities.} We set capacities based on the
fractional assignment returned by
Lemma~\ref{lem:fractional-assignment}. Let $\{X_{j,t}\}$ denote this
fractional assignment. Then, we set $\Btl$ to be equal to $\sum_{j:
  l_j=l} q_j X_{j,t}$ plus a reserve capacity of $\epsilon' B_t$ where
$\epsilon'=\epsilon/\lmax^2$. It is immediate that the per-slot
capacity constraints are satisfied: for all $t\in \integers_+$, we
have,
$$ \sum_{l\in[l_{\max}]} \sum_{t'\in [t-l+1, t]} \Btl[t'] \le \sum_j
\sum_{t'\in [t-l_j+1, t]} q_j X_{j,t} + \epsilon' B_t \lmax^2 \le
(1-\epsilon)B_t + \epsilon B_t = B_t.$$ Furthermore, the fraction
assignment of Lemma~\ref{lem:fractional-assignment} gives a
$(1-\epsilon)$-approximation to social welfare while respecting the
block-wise capacity constraints in expectation.

\paragraph{Network over time blocks.} We will think of the graph over
time blocks as partitioned into $\lmax$ layers, with layer $\layerl =
\{(t,l)\}_{t\in\integers_+}$ corresponding to all blocks of length
$l$. Within each layer, the induced subgraph is a graph over
(starting) time slots. Each job's preference ordering, restricted to
layer $\layerl$, is identical to the preference ordering induced in
the unit-length case when slot prices are given by
$p_t(l)$. Accordingly, we define a network $D_l$ over $\layerl$ in a
manner analogous to the definition of network $D$ in
Section~\ref{sec:temporal} with respect to prices $\{p_t(l)\}$: $D_l =
(\layerl, E_{F, l}\cup E_{B, l})$. Finally, let $E_L = \{((t,l),
(t,l+1))\}_{t\in\integers_+, l\in[\lmax-1]}$ denote ``inter-layer'' edges
that go from each block $(t,l)$ to block $(t,l+1)$. Let $\DD = (\cup_l
\layerl, \cup_l (E_{F, l}\cup E_{B, l})\cup E_L)$.

Observe that the network $\DD$ has maximum in-degree $4$. We now argue
that the realized path $P_j$ of each job $j$ can be ``short-cut'' into
a path in the graph $\DD$. Suppose that the realized path of a job $j$
of length $l_j$ starts at block $(y_j, l_j)$ and terminates at block
$(z_j, l)$ for $l\ge l_j$. Observe that if $l>l_j$, prior to
considering block $(z_j, l)$, the job must have considered every block
$(z_j,l')$ with $l'\in [l_j, l-1]$; all of these blocks $(z_j,l')$ are
in $P_j$. Now, define the path ${\tilde P}_j$ in two parts as
follows. The first part is a short-cut of the prefix of $P_j$ from
$(y_j, l_j)$ to $(z_j, l_j)$ defined over the layer $\layerl[l_j]$ as
in Section~\ref{sec:temporal}. The second part is a sequence of
inter-layer edges connecting $(z_j,l')$ to $(z_j, l'+1)$ for
$l'\in [l_j, l-1]$. 

It is easy to see that ${\tilde P}_j$ is a short-cut of $P_j$ and lies
in the graph $\DD$.  Corollary~\ref{lem:short-cut} then implies that
the collection of realized paths $\realPaths'=({\tilde P}_j)$
satisfies the min-work condition and Theorem~\ref{thm:Network-main}
can be applied. It remains to argue that for every block $(t,l)$, the
moment generating function of $A_{t,l}-\Btl$ is
bounded, 
where $A_{t,l}$ is the random number of fresh arrivals at the
block. Recall that $\Btl = \sum_{j: l_j=l} q_j X_{j,t} + \epsilon'
B_t$ where $\epsilon'=\epsilon/\lmax^2$. On the other hand, $A_{t,l} =
\sum_{j: l_j=l} q_j \hat{X}_{j,t}$, where $\hat{X}_{t,l}$ is a
Bernoulli variable with expectation $X_{j,t}$. So, we have
$$\EE { e^ {\epsilon' (A_{t,l}- \Btl)}} \le e^{-\frac 12\epsilon'^2
  \Btl} \le e^{-\frac 12\epsilon'^3 B_t} \le \epsilon^{c/2} $$ which for
an appropriate constant $c$ is at most $\epsilon^2/4e$. Here the
second inequality used the fact that $\Btl\ge\epsilon' B_t$, and the
third used the lower bound on $B_t$ from the statement of
Theorem~\ref{thm:stability-cloud}.  Therefore,
Theorem~\ref{thm:Network-main} applies and each job is accepted with
probability at least $1-\epsilon'$. We achieve an approximation factor of
$(1-\epsilon')(1-\epsilon)\ge 1-2\epsilon$ for social welfare.

This concludes the proof of Theorem~\ref{thm:stability-cloud}. 

\paragraph{Truthfulness and job payments.} Truthfulness of the above
mechanism is straightforward to argue: each job is allocated the cheapest block
available that meets its requirements at the time of its
arrival. Observe that a job of length $l_j$ that is allocated block
$(t,l)$ for some $l>l_j$ must pay the price $p_t(l)$ (and not the
cheaper price $p_t(l_j)$) in the above mechanism. It is, however,
possible to modify our argument so that the theorem holds also when a
job of length $l_j$ can buy a slot $(t,l)$ with $l>l_j$ at a price of
$p_t(l_j)$. This change to the mechanism changes each job's
preference ordering and realized path, but realized paths can once
again be short-cut to form paths in $\DD$, and we obtain the same
conclusion as before. Finally, the new mechanism continues to be
truthful with respect to jobs' lengths: a job paying $p_t(l_j)$ for some
block $(t,l)$ with $l>l_j$ is terminated after $l$ steps, so it hurts
to report a length smaller than the true length.
\subsection*{Acknowledgements}
We are grateful to TJ Gilbrough for making almost all of the figures in this paper.
\bibliographystyle{plainnat}
\bibliography{queueing}
\newpage

\appendix
\section{Related Work} 
\label{sec:related-full} 

\paragraph{Mechanism design for online allocation.} In online settings, a mechanism is called \emph{prompt} if the payments are computed as soon as the job is scheduled, and is {\em tardy} if payments are computed at some later point in time (usually after the deadline of a job).
The truthful online scheduling problem has been extensively studied in the worst case competitive analysis framework. 
	Lavi and Nisan \cite{lavi2015online} introduced the problem of truthful online scheduling for unit length jobs on a single machine, with the social welfare objective, and showed that no deterministic mechanism that is truthful w.r.t. all the parameters can get an approximation ratio $< T$, where $T$ is the time horizon. 
	They proposed a weaker notion of truthfulness that they call set-Nash, and gave constant competitive mechanisms satisfying set-Nash.  
	Hajiaghayi et al.~\cite{hajiaghayi2005online}  gave a tardy 2 approximation for unit length jobs. Their mechanism is truthful with the assumption of no early arrivals and late departures. They also extended this to an $O(\log \lmax)$ approximation for jobs of different lengths, where $\lmax$ is the ratio of the maximum to minimum length of a job. 
	Cole et al.~\cite{cole2008prompt} gave a prompt 2 approximation for unit length jobs, that is truthful only w.r.t. the value. 
	They extended it to a prompt $O(\log \lmax)$ approximation for different length jobs, that is truthful w.r.t the value and the deadline. 
	Azar and Khaitsin~\cite{azar2011prompt} designed a prompt mechanism for unit length jobs with arbitrary \emph{width}, on a single machine, that is a 6-approximation. The mechanism is truthful only w.r.t. the value. 
	In a more recent work, Azar et al.~\cite{azar2015truthful} assumed that there is a lower bound $s$ on the \emph{slack} of each job, which is the ratio of the length of the  [arrival, deadline] window to the job's length. They obtained a $2 + O(1/(\sqrt[3]s-1) ) + O(1/(\sqrt[3]s-1)^{3} )$ approximation for arbitrary length jobs, via a mechanism that is truthful w.r.t all the parameters, under the assumption of no early arrival, no late departure, and no under-reporting of length. The mechanism is tardy, but can be modified to make decisions earlier with further assumptions on the slack. 
	In the absence of slack, even algorithmically (i.e., with no truthfulness constraints), the online problem with arbitrary length jobs has a lower bound on the competitive ratio that is polylogarithmic in $l$ or $\mu$, where $\mu $ is the ratio of the largest to smallest possible values \citep{canetti1998bounding}. 

Other results: 
	Although stated in terms of combinatorial auctions,  the results of Feldman et al.~\cite{feldman2015combinatorial} are relevant. 
	They show how posted prices can achieve a truthful $2$-approximation in combinatorial auctions with XOS biddders, in the \emph{Bayesian} setting. This implies a $2$-approximation for unit length jobs that is truthful w.r.t. all the parameters. 
The algorithmic problems of stochastic online matching and generalizations, under large budgets/capacities, are similar in spirit to the stochastic process we consider \cite{GoelMehta,DH09,Devanur2011,agrawal2014dynamic,Feldman10,KTRV14,agrawal2015fast}.  The temporal aspects of the two problems are very different, due to which standard models in that literature such as the random order model are not a good fit here.

\paragraph{Connections to queueing theory.}
Our models are closely related to standard models in queuing theory, when the demand and supply are stationary (i.e., not changing with time). 
In particular, for unit length jobs, suppose that $B_t=B$ for all $t$, the advertised prices are all equal, and every client tries to obtain service at the first slot in its window, failing which it moves its demand to the next time slot, and then the next, and so on.
This case corresponds to the standard  M/D/B queueing model, with Markovian arrivals, deterministic processing time, and $B$ servers, under the first-come first-served (FCFS) queuing discipline.\footnote{The notation for different queuing models is as follows: an A/B/C queue is one where the inter job arrival times are drawn from distributions in family A, the job lengths distributions belong to family B, and there are C identical machines. D is the class of deterministic distributions, M is the class of exponential distributions, GI is the class of general, independent distributions, GIB is the same class with a bounded support, and an $H_2^*$ distribution is a mixture of an exponential and a point mass.} We consider the regime where the rate of arrival of total work is $B- O(\sqrt B)$. 
While we would like to analyze the probability of completion of any given job within its deadline, an easier quantity to compute, that is also an upper bound on this, is the probability that all $B$ machines are busy, which is called the delay probability. 
		This question was studied already by the seminal paper of Erlang~\cite{erlang1948rational}, which initiated the study of queues. 
		In particular Erlang's C model refers to an M/M/B queue and a closed form expression for the delay probability is derived. 
	The importance of this regime was recognized by Halfin and Whitt~\cite{halfin1981heavy}, and is now called the Halfin-Whitt regime or the Quality and Efficiency Driven (QED) regime. This is because in this regime one can hope for high efficiency, which refers to a utilization ratio close to 1, and high quality, which refers to a delay probability close to 0.  Halfin and Whitt~\cite{halfin1981heavy} gave a formula for the delay probability of GI/M/B queues in this regime, 
	and this was extended to $H_2^*$ distributions (a mixture of exponential and a point mass) by \cite{whitt2005heavy}. 
	Jelenkovi{\'c} et al.~\cite{jelenkovic2004heavy} did the same for GI/D/B queues. 
	
	Good bounds on the delay probability for more general job length distributions are not known. In particular, it is open whether the delay probability is 
	bounded above by $\delta$ for all job length distributions when the rate of work is at most $B - c \sqrt{B \log(1/\delta)} $ for some universal constant c.
	In fact it is not even known if such a bound holds for all distributions supported on $[0,L]$, when the rate of work is at most $B -  \text{poly}(L) \sqrt{B \log(1/\delta)} $. 
	Whitt~\cite{whitt2004diffusion}  gives heuristic approximations for the delay probability and other related quantities for  GI/GI/B queues, and  
	Psounis et al.~\cite{psounis2005systems}	do the same for heavy tailed distributions, using an expression derived from  a ``bimodal" distribution.  
	These are not proven theorems, but are rather shown to be good approximations via numerical analysis, or using simulations on traces of real workloads. 
	The state of the art in this area is by Goldberg~\cite{goldberg2013steady}, who gives bounds on the delay probability as a limit of limits: the limit as $c\rightarrow \infty$, and as a function of $c$, the limit as $B \rightarrow \infty$, of the delay probability of GI/GI/queues with arrival rate of $B - c \sqrt{B}.$ The convergence is not uniform, as the rate depends on both the arrival and job length distributions.

\section{Proofs for Section~\ref{sec:network}}
\label{sec:network-full}

\subsection{The tree setting: proof of Lemma~\ref{lemma:TreeProcess}}

We begin by proving that the stability of service theorem holds for the network of servers setting when the network is a tree.

\begin{numberedlemma}{\ref{lemma:TreeProcess}}
Fix an $\epsilon$ in $[0,1/2]$. Suppose that the network $G$ is a
finite directed tree, that is, it contains no directed cycles and every node has out-degree $1$, and that the moment generating function of $A_i-B_i$
for each node $i$ satisfies $\EE { e^ {\epsilon (A_i- B_i)}} \le
\epsilon ^2/e\dmax.$
Then, for any $i$,
$\Pr[{\ell_i(\realPaths)\ge B_i}] \le \epsilon$.
\end{numberedlemma}

\begin{proof}
  Recall that when the network is a tree, every node $v$, after
  processing the first $B_v$ jobs that arrive at this node, forwards
  all of the remaining jobs to its parent.  
We call a node a {\em leaf} if it has no incoming edges. Order the nodes in the tree in topological order starting from the leaves. Let $F_v$ be the number of jobs forwarded by node $v$ to its parent, and let
$$F'_v := \max(A_v+\sum_{i=0}^dF_{u_i}-(B_v -1), 0).$$
Clearly $\ell_v(\realPaths)\ge B_v$ if and only if $F'_v > 0$.
We will prove by induction over the topological ordering that $$\EE{ e^{\epsilon F'_v}} \le \rho, \quad\text{ where }\quad \rho = 1 + \frac{\epsilon ^2}{d}.$$
The base case is a leaf (and follows from the argument below).

\begin{figure}[!htbp]
\begin{center}
\includegraphics[width=0.3\textwidth]{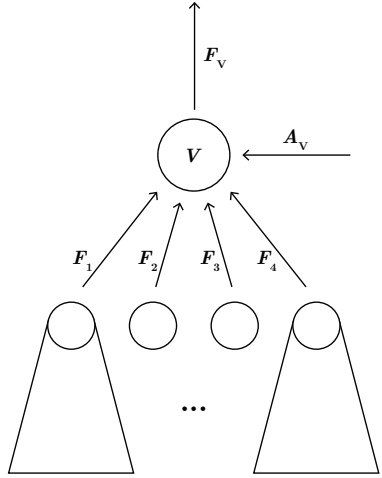}
\end{center}
\caption{
The induction step
}
\label{fig:VTree}
\end{figure}

For the induction step,
consider a node $v$ where up to $d$ predecessors are $u_1, \ldots, u_d$. (See Figure~\ref{fig:VTree}). By the induction hypothesis and
the fact that $F'_{u_i} \ge F_{u_i}$, we have
$$\EE{ e^{\epsilon F_{u_i} }} \le \rho$$
Note that $F_{u_i}$ and $F_{u_j}$ are independent for each distinct $i$ and $j$ since the trees rooted at them are disjoint, and they are also independent of $A_v$.
Thus, we have
\begin{align*}
\EE{e^{\epsilon F'_v}}  & \le \EE{e^{\epsilon \max(A_v+\sum_{i=0}^dF_{u_i}-B_v+1, 0)}}\\
& \le 1 + e^{\epsilon}\EE{e^{\epsilon (A_v+\sum_{i=0}^dF_{u_i}- B_v)}} \\
& \le 1 +  e^{\epsilon}\EE{e^{\epsilon (A_v- B_v)}}\rho^d\\
& \le 1 + \frac{e^{\epsilon}\epsilon ^2}{ed} \rho^d \le \rho,
\end{align*}
for $\rho$ defined as above. Here the last inequality follows by observing:
$$1 + \frac{e^{\epsilon}\epsilon ^2}{ed}\left(1 + \frac{\epsilon ^2}{d}\right)^d\le 1+ \frac{e^{\epsilon}\epsilon ^2}{ed} e^{d \ln \left(1 + \frac{\epsilon ^2}{d}\right)}\le
  1+ \frac{e^{\epsilon}\epsilon ^2}{ed} e^{\epsilon^2} \leq1 + \frac{\epsilon ^2}{d},$$
 as long as $\epsilon + \epsilon^2 \le 1$.
 Letting
$$\eta_v := \pr{F'_v >  0},$$
and recalling that $F'_v$ is integral, we have,
$$ 1-\eta_v + \eta_v e^{\epsilon} \le \EE{e^{\epsilon F'_v}}  \le \rho.$$
Solving for $\eta_v$, we obtain
$$\eta_v(e^{\epsilon}-1) \le \rho -1 = \frac{\epsilon ^2}{d}$$
so 
$$  \eta_v \le \frac{\epsilon ^2}{ d(e ^{\epsilon}-1)}  \le  \frac{\epsilon 
^2}{d\epsilon} \le \frac{\epsilon}{d}\le \epsilon.$$
Therefore, 
$$\pr{\ell_v(\realPaths)\ge B_v} = \pr{F'_v >  0} \le \epsilon.$$
\end{proof}

\subsection{Reducing to a tree: proof of Lemma~\ref{lem:step1-Tree}}

We will now prove that for every instantiation of arrivals and forwards in the network of servers setting on $G$ and every node $u$, we can find a subtree $T$ of $G$ rooted at $u$, such that the load at $u$ becomes worse when the process is run over the tree $T$. Before we restate the main result of this section, let us recall some notation. Let $a_i$ denote the realized number of  jobs arriving at node $i$ in $G$, and $\ba = (a_i)$; let $\realPaths$ denote the realized paths of jobs. Let $\Trees{u}$ denote the set of all directed subtrees of $G$ rooted at node $u$, and for $T\in \Trees{u}$, let $\ell_i^T(\ba)$ denote the load on node $i$ given the realized arrivals $\ba$, when jobs that have not yet been served are routed along the tree $T$. (See Figure~\ref{fig:Routing}.)

\begin{numberedlemma}{\ref{lem:step1-Tree}}
If a fixed set of arrivals (resulting in a particular $\ba$) and realized paths $\Paths$ overload a
node $u$ in the network $G$, then  $\exists$ a tree $T\in\Trees{u}$
such that $u$ is overloaded with the same set of arrivals when
jobs are routed along $T$. Formally, 
$$\textstyle  \ell_u (\Paths) \ge B_u \quad\text{ implies that }\quad \max_{T\in
  \Trees{u}} \ell_u^T (\ba) \ge B_u.$$
\end{numberedlemma}

The proof of Lemma \ref{lem:step1-Tree} proceeds in several steps.

\subsubsection*{Step 1: Remove cycles}

Throughout the argument we will progressively modify the realized paths of jobs, while maintaining the invariant that every node $i$ must process at least $B_i$ jobs before forwarding any jobs. To this end, we say that a set of paths $\realPaths'=(P'_j)$ is {\bf valid for arrivals $\ba$} if there is 
an ordering of arrival and forwarding events
for the realized jobs consistent with the arrivals $\ba$, such that the realized path of each job $j$  is exactly $P'_j$, and $P'_j$ is a path in $G$.

For a directed multi-graph $G'$, let $\inn_{G'}(i)$ and $\outt_{G'}(i)$ denote the in- and out-degrees, respectively, of node $i$ in the multi-graph.  We first show that a set of paths $\realPaths'$ is valid for $\ba$ if and only if the multi-graph given by the union of the paths, call it $G'$, satisfies the following {\bf min-work condition}: \begin{equation}
 \label{minWork}
\forall i, \quad \outt_{G'}(i) \le \max (0, \inn_{G'}(i) + a_i - B_i).
\end{equation}

\begin{claim}
\label{claim:Cycles}
A multi-graph $G'$ can be decomposed into set of paths that is valid for arrivals $\ba$ if and only if it satisfies the min-work condition~\eqref{minWork}.
\end{claim}

\begin{proof}
The ``only-if'' direction of the statement follows trivially from the definition of valid paths. For the ``if'' direction, define
\begin{equation}
\label{Conservation}
\di := a_i + \inn_{G'}(i) - \outt_{G'}(i),
\end{equation} 
which is the number of departures at node $i$. (These departures can occur either because a job is processed at $i$ or because its path terminates.)
Clearly $\sum_i a_i = \sum_i \di$.
Construct an $s-t$ flow network, where there is an edge of capacity $a_i$ from $s$ to node $i$,
an edge of capacity $\di$ from node $i$ to $t$, and an edge
from $i$ to $j$ of capacity equal to the number of edges in $G'$ from $i$ to $j$.
Clearly, there is an integer flow in which each edge is filled to capacity, and therefore,
this flow can be decomposed into paths. 

Moreover, the paths are trivially consistent with the property that each node processes $\min(\di, B_i)$ jobs, since by \eqref{minWork}, any $i$ with positive outdegree satisfies
$$0 < \outt_{G'}(i) \le \inn_{G'}(i)  + a_i - B_i,$$
or in other words, $$\di = a_i + \inn_{G'}(i) - \outt_{G'}(i) \ge B_i.$$
\end{proof}



\begin{cor}
 Let $\ba = (a_i)$ be a set of arrivals and $\realPaths$ a set of paths that are valid for $\ba$. Let $G'$ be the multigraph obtained by taking the union of paths in $\realPaths$.  Successively remove directed cycles from $G'$ to obtain a new acyclic multigraph $G''$. Then $G''$ can be decomposed into a set of paths $\realPaths'$ valid for the arrivals
$\ba$, with the property that a node $i$ is overloaded under
$\realPaths$ if and only if it is overloaded under $\realPaths'$.
That is, for all $i$,
\begin{equation}
\label{eqn:cycle}
\ell_i (\realPaths) \ge B_i \longrightarrow \ell_i(\realPaths') \ge B_i 
\end{equation}
and
$$\ell_i (\realPaths) <  B_i \longrightarrow \ell_i(\realPaths') < B_i.$$
\end{cor}
\begin{proof}
Prior to the removal of cycles, the  paths $\realPaths$ satisfied 
 the property \eqref{minWork}.
Since removing a cycle preserves the min-work property \eqref{minWork} and the resulting graph is acyclic, by Claim \ref{claim:Cycles}, its edges decompose into a valid set of realized paths. This is the set of paths $\realPaths'$. See Figure \ref{fig:Evolution}.

\begin{figure}[htbp]
\begin{center}
\includegraphics[width=\textwidth]{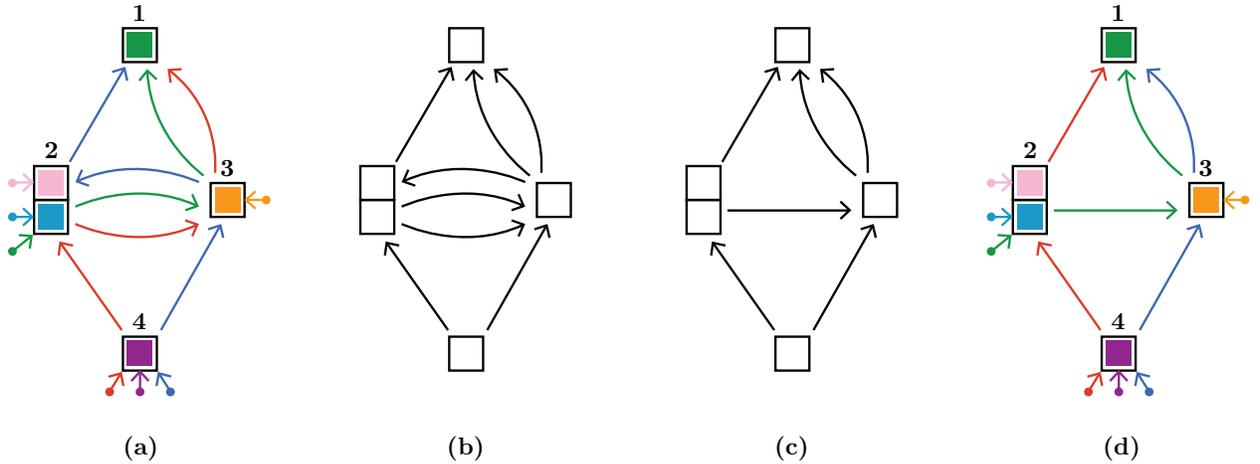}
\end{center}
\caption{
Figure (a) shows a set of arrivals and paths taken. Figure (b) is the induced multigraph $G'$. Figure (c) shows the multigraph $G''$
obtained after removing cycles from $G'$ (arrivals are not shown). Figure (d) 
shows  a decomposition of $G''$ into set of valid paths, as per
Claim \ref{claim:Cycles}. (This decomposition is not unique.) All overloaded
nodes are still overloaded.}
\label{fig:Evolution}
\end{figure}

Clearly the in-degree of each node is weakly decreasing. Furthermore,
nodes with out-degree $=0$ do not belong to any cycles, and so their
loads don't change. So we only need to show
that  (\ref{eqn:cycle}) holds for nodes with out-degree $>0$ after removing cycles.
This follows from the fact that the out-degree
of a node in $G'$ is an upper bound on the number of cycles removed that it is part of. By (\ref{minWork}), $\outt_{G'}(i) \ge k$
implies that $\inn_{G'}(i) + a_i \ge B_i+k$. 

Therefore, if  $k$ cycles through queue $i$ are removed during the process of removing cycles, the final load at $i$ (in-degree plus arrivals) is at least $B_i$, so (\ref{eqn:cycle}) holds.
\end{proof}

Via a similar argument, we also obtain the following corollary  that
will be used in the proof of Lemma \ref{lem:path-reduction-2}.
It allows us to ``short-cut'' paths without affecting overload events. 

\begin{cor}
\label{lem:short-cut}
Let $\ba = (a_i)$ be a set of arrivals and $\realPaths$ a set of
realized paths valid for $\ba$. Consider a path $P\in\realPaths$ of
length at least $2$, and let $(u_1, u_2)$ and $(u_2, u_3)$ be two
consecutive edges in this path. Let $P'$ be obtained by removing
(short-cutting) the vertex $u_2$ from $P$. That is,
$P'=P\setminus\{(u_1, u_2), (u_2, u_3)\}\cup\{(u_1, u_3\}$. Then, the
new set of paths $\realPaths'=\realPaths\setminus\{P\}\cup\{P'\}$ is
valid for $\ba$. Furthermore, for all nodes $i$, $\ell_i
(\realPaths)\ge B_i$ iff $\ell_i (\realPaths')\ge B_i$.
\end{cor}

\subsubsection*{Step 2: Modify paths to obtain tree}

For the rest of this subsection, we will assume that we are given the
arrivals $\ba$, a set of valid paths $\realPaths$ that form a directed
acyclic graph, and a specific node $u$. Let $G'$ denote the
multi-graph formed by taking the union of the paths in
$\realPaths$. We complete the proof of Lemma \ref{lem:step1-Tree} by
modifying the paths so that they are directed along a tree rooted at
$u$, without decreasing the load on $u$. The modified paths will remain valid.

To this end, we will repeatedly use the following two operations:
\vspace{0.1in}

\noindent
{\em Operation 1:} Remove an edge  $(i,j)$ if $j$ has out-degree 0, and $j\ne u$.

\vspace{0.1in}
\noindent
{\em Operation 2: }  Suppose that $P_1$ and $P_2$ are two edge-disjoint paths that
start at $i$ and end at $j$, and there is a path from $j$ to $u$.
Delete path $P_2$ and replace it by a duplicate copy of $P_1$.

\vspace{0.1in}

\begin{figure}[htbp]
\begin{center}
\includegraphics[width=0.9\textwidth]{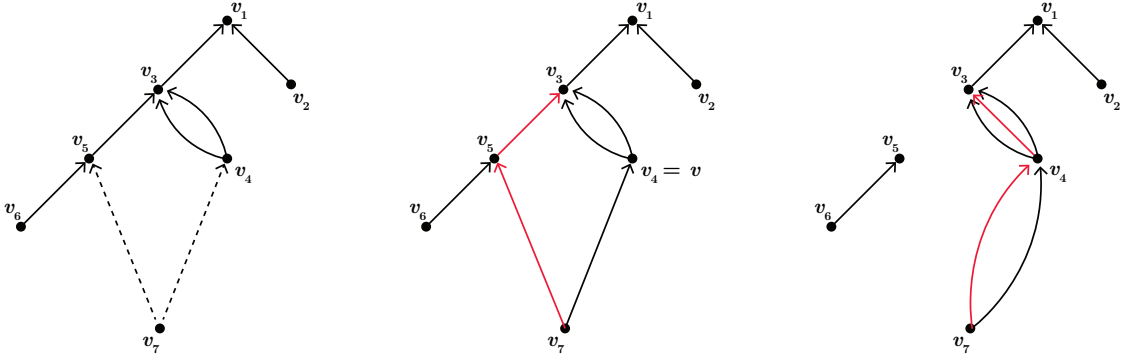}
\end{center}
\caption{
This figure shows an example of the application of inductive step when $i=7$. Initially (left figure) all paths from $v_1, \ldots, v_6$ form a tree directed towards $v_1$, and we are about to process $v_7$ which has edges to $v_4$ and $v_5$. Operation 2 is applied to the paths $P_1 = (v_7, v_4, v_3)$ and $P_2 = (v_7, v_5, v_3)$ resulting in the graph shown on the right. Subsequently, an application of operation 1 removes the edge $(v_6, v_5)$.  }
\label{fig:DirectedTree}
\end{figure}

\begin{figure}[htbp]
\begin{center}
\includegraphics[width=0.8\textwidth]{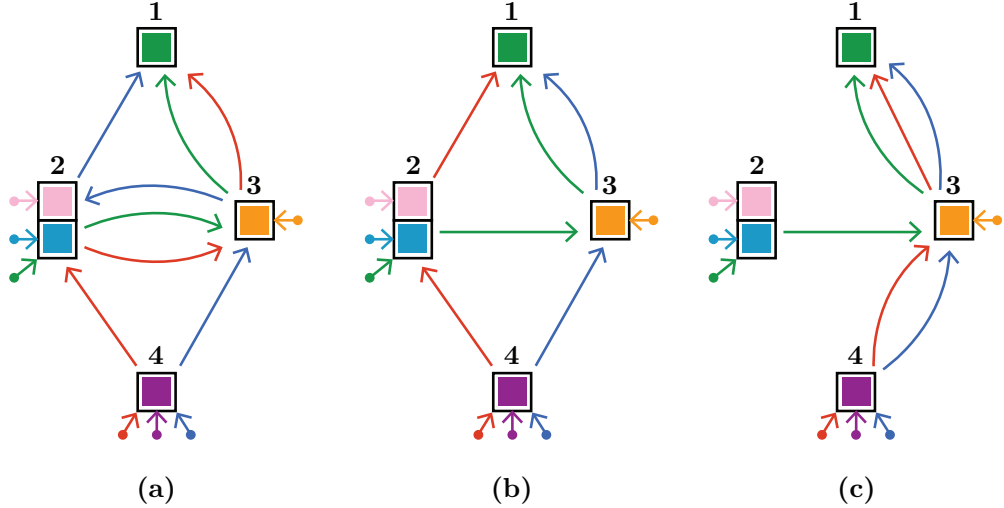}
\end{center}
\caption{
This figure shows the transformations applied to convert the paths into valid
paths along a tree (in this case tree $T_3$ from Figure \ref{fig:Routing}).
Going from the left graph to the middle is the result of removing cycles.
The right figure shows the paths obtained once we apply Step 2, which modifies paths to obtain a tree. In this example, the red path was rerouted to go through 3 rather than 2. The load at node 1 is preserved in the transformation from the middle routing to
the routing on the right. }
\label{fig:ColorfulRouting}
\end{figure}

If we begin with a set of paths valid for $\ba$, then by Claim
\ref{claim:Cycles}, operations 1 and 2 preserve the existence of a set
of realized paths that are valid for the arrivals $\ba$. Indeed,
operation 1 reduces the out-degree of a node. Operation 2 has the
following properties:
\begin{itemize}
\item It reduces both the in-degree and out-degree of some nodes by 1 (every node on $P_2$ except for $i$ and $j$),
which preserves \eqref{minWork}.
\item It increases both the in-degree and out-degree of some nodes by
  1 (nodes on $P_1$ except for $i$ and $j$), but only nodes that already
had out-degree 1, which also preserves \eqref{minWork}.
\item It maintains the out-degree of $i$ and in-degree of $j$.
\end{itemize}

We apply these two operations to get our tree as follows: Recall that
$G'$ is acyclic, and consider the nodes in $G'$ in topological order
(from sinks to sources), say $v_1, \ldots, v_n$.  Let $S_i := \{v_1,
\ldots, v_i\}$. We inductively apply the above operations so that the
subgraph on $S_i$ consists of a collection of paths terminating at
$u$, for which the corresponding graph (not multigraph) is a tree
directed towards and rooted at $u$. (This tree could be empty.)

The base case is $i=0$ (or the empty set).
To extend from $S_{i-1}$ to $S_i$,
we do the following: If all of the out-edges from $v_i$ are to nodes with out-degree 0, then remove
all of these edges (applying operation 1).
Otherwise, suppose that $v_i$ has an edge to some vertex $v\in S_{i-1}$
from which there is a path $P$ to $u$. Pick such a $v$, and the associated path $P$.
Repeat the following two steps until $S_i$ satisfies the inductive hypothesis:
\begin{enumerate}
\item As long as there is an  edge $(v', v'')$ with $v'' \in S_i\setminus u$,  and $\outt_G (v'') =0$, apply operation 1 to remove
this edge.
\item If there is an edge $(v_i, v')$ where $v'\ne v$ and there is a path $P'$ from $v'$ to $u$, find the first node $j$ at which the paths $P$ and $P'$ intersect. Let $P_{v,j}$ be the prefix of $P$ terminating at $j$, and let $P_{v',j}$ be the prefix of $P'$ terminating at $j$. Apply operation 2 to the paths $P_1 := (v_i, P_{v,j})$,and $P_2 := (v_i, P_{v',j})$. See Figure \ref{fig:DirectedTree}.
(Note that if there is no path $P'$ from $v'$ to $u$, then by the inductive hypothesis, $v'$
has out-degree 0, which means the edge $(v_i, v')$ will be removed.)
\end{enumerate}

This argument completes the proof of Lemma \ref{lem:step1-Tree}.

\subsection{The decorrelation lemma}

We will now prove the decorrelation lemma that is needed in our
analysis. We use the notation $\stgeq$ to denote stochastic dominance.

\begin{numberedlemma}{\ref{lem:probmore}}
{\bf (Decorrelation lemma for the max function)}
	Let $g_{\ell}:\Re \rightarrow \Re$, $\ell=1, \ldots, k$ be any non-decreasing functions,  $X$ be any real valued random variable, and let $Y_1, \ldots, Y_k$ be independent and identically distributed random variables from the same distribution as $X$.
	Then, 
	\[ \textstyle \max_{\ell}\{g_{\ell} (Y_{\ell})\} \stgeq \max_{\ell}\{ g_{\ell}(X)\}  .\]
\end{numberedlemma}

\begin{proof}
	Since the $Y_{\ell}$'s are independent, for any $a$, 
\begin{align}
\label{decorr-1}
\Pr[\max_{\ell}\{g_{\ell} (Y_{\ell})\}\leq a ] & =
\prod_{\ell}\Pr[g_{\ell}(Y_{\ell}) \leq a],\\
\intertext{whereas, recalling that $g_{\ell}$ is non-decreasing for each $\ell$, and setting $x^*(a) = \min_{\ell} g_{\ell}^{-1}(a)$,\footnote{Here, $g_{\ell}^{-1}(a) = \sup\{x : g_{\ell}(x)\le a\}$.} }
\label{decorr-2}
\Pr[ \max_{\ell}\{ g_{\ell}(X)\} \leq a ] & = \Pr[X \leq x^*(a)] = \min_{\ell} \Pr[g_{\ell}(X) \leq a].
\end{align}
	Clearly, the RHS of \eqref{decorr-2} is larger for all $a$
        than the RHS of \eqref{decorr-1}. Therefore, the lemma follows. 
\end{proof}

We can further generalize the decorrelation lemma as follows.

\begin{lemma}
	\label{lem:probmoreFull}
	Let $h_k:\Re ^n \rightarrow \Re $, $k=1, \ldots, N$ be functions that are non-decreasing in each variable. Let $A_1, \ldots, A_n$ be independent (but not identically distributed random variables) and,
	for each $1\le i \le n$, let
	$P_i: [N] \rightarrow [n_i]$, for nonnegative integers $n_i$. Then, 
	\begin{equation}
	\label{FullSI}
	 \max_{1\le k  \le N}\{ h_k(\{A_{i, P_i(k)}\}_{i=1}^n)\}\stgeq\max_{1\le k  \le N}\{h_k (\{A_i\}_{i=1}^n)\} .
	 \end{equation}
	 where each $A_{i, P_i(k)}$ is an independent draw from the distribution of $A_i$.
\end{lemma}
\begin{proof}
We prove by induction on $j$ that
$$	 \max_{1\le k  \le N}\left[h_k\left(\{A_{i, P_i(k)}\}_{i=1}^j, \{A_i\}_{i=j+1}^n\right)\right]\stgeq\max_{1\le k  \le N}\{h_k\left(\{A_{i, P_i(k)}\}_{i=1}^{j-1}, \{A_i\}_{i=j}^n\right)\} .$$
The base case of $j=0$ is immediate.
For the induction step, condition on all variables other than $A_j$ (which are independent of $A_j$), and define
$$h'_k( X) := h_k\left(\{A_{i, P_i(k)}\}_{i=1}^{j-1},X, \{A_i\}_{i=j+1}^n\right) | \{A_{i, P_i(k)}\}_{i=1}^{j-1}\text{ and } \{A_i\}_{i=j+1}^n.$$
Thus, it suffices to show that
\begin{equation}
\label{eqn:Gen}
\max_{1\le k  \le N}\left[h'_k\left(A_{j, P_j(k)}\right)\right]\stgeq\max_{1\le k  \le N}\{h'_k\left(A_j\right)\} .
\end{equation}
Letting
$$f_{\ell} (A_j) := \max_{1\le k \le N\text{ s.t. }P_j(k) = \ell} h'_k (A_j),$$
showing (\ref{eqn:Gen})
is the same as showing that
$$\max_{1\le \ell \le n_j} f_{\ell} (A_{j,\ell}) \stgeq \max_{1 \le \ell\le n_j} f_{\ell} (A_j),$$
which follows directly from Lemma \ref{lem:probmore}.
\end{proof}

\subsection{Reducing to a tree of trees: proof of Lemma~\ref{SDApp}}

\begin{numberedlemma}{\ref{SDApp}}
For each $i$ and directed simple path $P$ from $i$ to $u$, let
$A_{i,P}$ be an independent draw from the distribution of $A_i$, 
$P_i(T)$ be the unique path from $i$ to $u$ in tree $T$, and
$\stgeq$ denote stochastic dominance. Then
$$\textstyle \max_{T \in \Trees{u}}  \ell_u^T \left((A_{i, P_i(T)})_{i \in T}\right)\stgeq\max_{T \in \Trees{u}}\ell_u^T (\bA).$$

\end{numberedlemma}

\begin{proof}
  Apply Lemma \ref{lem:probmoreFull} with $N$ equal to the number of distinct trees $T$ rooted at $u$,
$n$ equal to the number of queues in the queueing network,
$h_k(\bA):= \ell_u^{T_k}(\bA)$, i.e., the load on $u$ when the each job follows the routes given by tree $T_k$ until it is processed, and $n_i$ equal to the number of distinct simple paths from $i$ to $u$.
\end{proof}
\section{Proofs for Section~\ref{sec:temporal}}
\label{sec:temporal-full}

\subsection{Properties of the graph $D$}

We begin by proving properties of the directed graph $D$. Recall that $D$ 
contains three types of edges. For every time slot $t\in \integers_+$, the 
``left forward'' edges connect $t$'s left parent $\ell(t)$ to $t$; the ``right 
forward'' edges connect $t$'s right parent $r(t)$ to it. The ``backward'' edges 
are defined as $E_B: = \{ (b(t), t) \quad \forall t\in T\}$
where $b(t)= \min \{s>t: p_s= p_t\}$. Recall that every node $t\in \integers_+$ 
in this graph has in-degree at most $3$. (See Figure \ref{fig:BoxesRedBlue}.)

Recall that $\Anc(t)$ denotes the ancestors of time slot $t$ in $\tilde D$ i.e. $\Anc(t) = \{s \text{ such that there is a path in }\tilde D\text{ from } s\text{ to }t\}.$

\begin{lemma} 
\label{lem:PropertiesAncestors}
\begin{enumerate}
\item[(a)] For any $t$, either there is an edge in $\tilde D$ from $\ell(t)$ to $r(t)$ or vice versa.  That is, either $\ell(t) = \ell(r(t))$ or $r(t) = r(\ell(t))$.
\item[(b)] For any $t$, the set $\Anc(t)$ is totally ordered: if $t_1,t_2 \in \Anc(t)$, then either $t_1 \in \Anc(t_2)$ or $t_2\in \Anc(t_1)$.
\end{enumerate}
\end{lemma}

\begin{figure}[htbp]
\begin{center}
\includegraphics[width=0.8\textwidth]{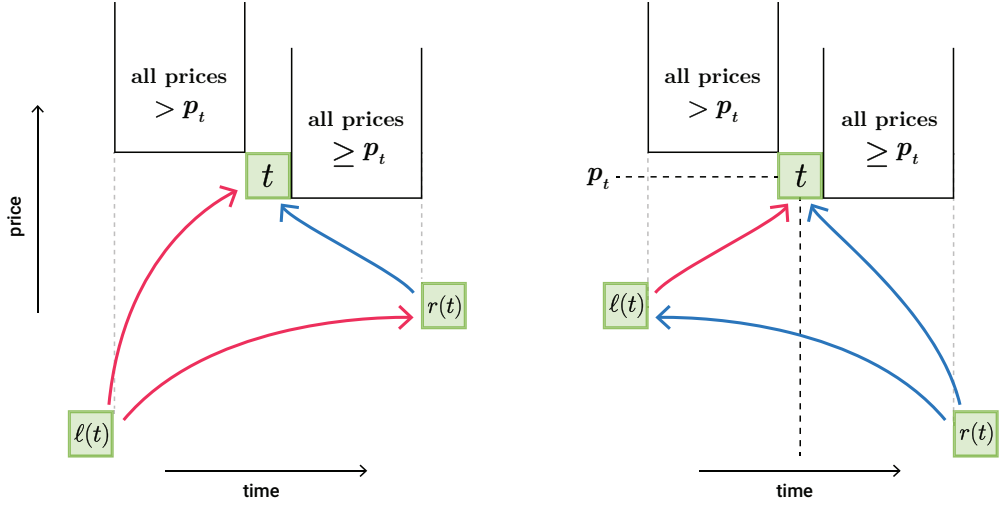}
\end{center}
\caption{
For a particular time slot $t$, $\ell(t)$ is the largest time less than $t$ where the price is at most $p_t$, and $r(t)$ is the smallest time greater than $t$ at which the price is strictly less than $p_t$. Notice that if $p_{\ell(t)} \le p_{r(t)}$, then by definition $\ell(t) = \ell (r(t))$. The right hand side shows the case where $p_{\ell(t)} > p_{r(t)}$. When this happens, then by definition $r(t) = r (\ell(t))$.}
\label{fig:BoxesGreater4}
\end{figure}

\begin{proof}
\begin{enumerate}
\item[(a)] See proof in caption of Figure \ref{fig:BoxesGreater4}.
\item[(b)] Proof by induction on $k = |\Anc(t)|$. For all $t$ with $|\Anc(t)| =1$, the statement is immediate. For $k > 1$,  suppose wlog that  $\ell(t) = \ell(r(t))$. Then all  ancestors of $t$ (other than $t$ itself) are ancestors of $r(t)$, and therefore, the total order on $\Anc(t)$ must terminate with $r(t)$ followed by $t$. Applying the induction hypothesis to $\Anc(r(t))$ completes the argument. See Figure~\ref{fig:Shortcutting} for an example of the total ordering.
\end{enumerate}
\end{proof}

\subsection{Properties of the paths ${\tilde P}_j$: proof of Lemmas~\ref{lem:path-reduction-1} and \ref{lem:path-reduction-2}}

Recall that we define job $j$'s new path ${\tilde P}_j$ as follows. Let $y_j$ and $z_j$ denote the first and last nodes on the job's realized path $P_j$. Recall that $P_j^1$ denotes the prefix of $P_j$ which visits nodes $t\in \favj$ with $t<y_j$, starting at $y_j$, and $P_j^2$ denotes the remaining suffix of the path. We begin with a simple property of the suffix $P_j^2$.

\begin{figure}[htbp]
\begin{center}
\includegraphics[width=0.4\textwidth]{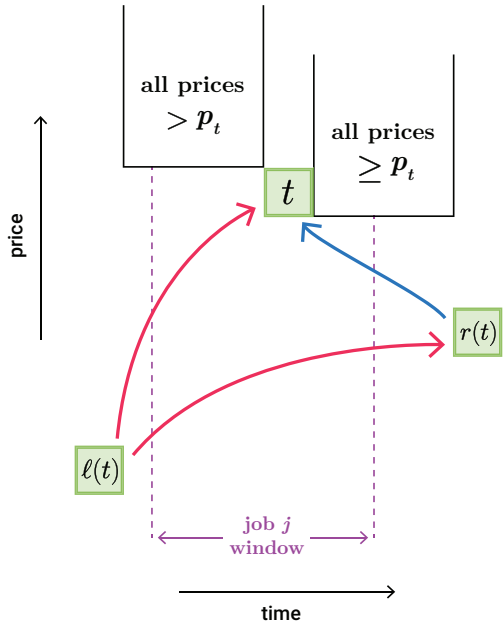}
\end{center}
\caption{
This figure illustrates the proof of Lemma \ref{lem:path-prop-1}.}
\label{fig:PricesGreater2}
\end{figure}

\begin{lemma}
\label{lem:path-prop-1}
  Let $t$ be a node in the path $P_j^2$ for some job $j$. Then, either $\ell(t)$ or $r(t)$ belongs to the job's window $W_j$, and appears before $t$ on the realized path $P_j$.
\end{lemma}

\begin{proof}
  If neither $\ell(t)$ or $r(t)$ are in job $j$'s window, $W_j=[s_j, d_j]$, then all prices in the window are at least $p_t$, and all prices in $[s_j, t-1]$ are strictly larger than $p_t$. See Figure \ref{fig:PricesGreater2}. Since the path $P_j$ starts at a cheapest slot in the window, it must start at a slot of price $p_t$, at a time $t$ or later. That is, $t\in\favj$ with $y_j\ge t$.
If the path doesn't start at $t$ itself, then by the definition of $P_j^1$, $t$ belongs to $P_j^1$. Finally, if one of $\ell(t)$ or $r(t)$ lies in $W_j$, then it is easy to see that job $j$ must visit this slot before it visits $t$: either the price at this slot is smaller than that of $t$, or, in the case of $\ell(t)$, the prices are the same, but the job visits slots to the left of $t$ with prices $p_t$ before $t$.
\end{proof}

We now complete our description of the reduction from the temporal setting to the network setting by specifying the paths ${\tilde P}_j$ for each realized job $j$. Recall that we define ${\tilde P}_j^2 = P_j^2\cap \Anc(z_j)$ to be a ``short-cutting'' of the suffix $P_j^2$. Let $s_1$ be the first node on the path ${\tilde P}_j^2$. By Lemma~\ref{lem:path-prop-1}, one of the parents of $s_1$ lies in $P_j$. Call this parent $s_0$. Since $s_0$ appears before $s_1$ in $P_j$, it must be the case that $s_0\in P_j^1$. Let ${\tilde P}_j^1$ be the prefix of $P_j^1$ from $y_j$ to $s_0$. Define ${\tilde P}_j=\tilde P_j^1\cup {\tilde P}_j^2$ if $P_j^2$ is non-empty, and ${\tilde P}_j=P_j^1$ otherwise.

We will now prove that the new paths ${\tilde P}_j$ lie in the graph $D$. (See Figure \ref{fig:Shortcutting}.)

\begin{figure}[htbp]
\begin{center}
\includegraphics[width=0.5\textwidth]{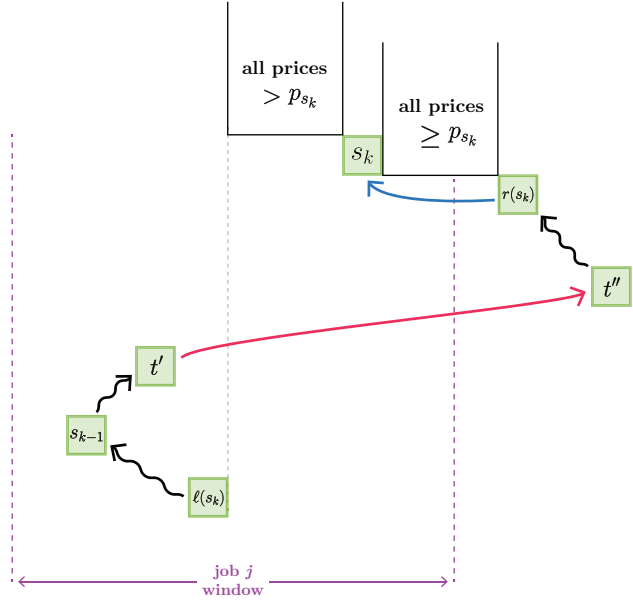}
\end{center}
\caption{
This figure illustrates the contradiction in the proof of 
Lemma \ref{lem:path-reduction-1}. Time $t'$ is less than time $\ell(s_k)$
and $p_{t'} > p_{\ell(s_k)}$. This contradicts the fact that $t' = \ell(t'')$.}
\label{fig:PricesGreater4}
\end{figure}

\begin{numberedlemma}{\ref{lem:path-reduction-1}}
Paths ${\tilde P}_j$ as defined above lie in the graph $D$.
\end{numberedlemma}

\begin{proof}
  Consider a job $j$. The path $\tilde P_j$ consists of three components: (1) a prefix of the path $P_j^1$, (2) the edge $(s_0,s_1)$, where $s_1$ is the first node on $\tilde P_j^2$ and $s_0$ is its ancestor on $P_j^1$, and, (3) the path ${\tilde P}_j^2 = P_j^2\cap \Anc(z_j)$. Observe that the edges in path $P_j^1$ are all backward edges. Therefore, they lie in the graph $D$. The edge $(s_0,s_1)$ lies in $\tilde D$ by construction, and by recalling that Lemma~\ref{lem:path-prop-1} implies that one of the parents of $s_1$ lies in $P_j^1$. We will now focus on the path ${\tilde P}_j^2 = P_j^2\cap \Anc(z_j)$. Let ${\tilde P}_j^2 = \{s_1, s_2, \cdots, s_k\}$, for some $k\ge 1$ where $s_k=z_j$.  We claim that for all $i\in \{1, \cdots, k-1\}$, $(s_i,s_{i+1})$ is an edge in $\tilde D$.

  We prove the claim by induction on the length $k$ of $\tilde P_j^2$: Consider the last node $s_k$ in $\tilde P_j^2$. Suppose, without loss of generality, that $\ell(s_k)= \ell(r(s_k))$.  We will prove that either $s_{k-1} = \ell(s_k)$ or $s_{k-1}=r(s_k)$.  Once this is proved, we simply use the fact that there is an edge $(s_{k-1}, s_k)$ in $\tilde D$, and complete the argument by applying the inductive hypothesis to $(s_1, \ldots, s_{k-1})$.

  Suppose then that $s_{k-1}$ is not $r(s_k)$ or $\ell(s_k)$. Since all ancestors of $s_k$ except $s_k$ itself are ancestors of $r(s_k)$, and $P_j$ doesn't go through $r(s_k)$, it must be that $r(s_k)$ is outside the job's window $W_j$. Thus, we have:
\begin{itemize}
\item $s_{k-1}$ is an earlier time than $\ell(s_k)$: this follows from the fact that all prices in $[\ell(s_k)+1, r(s_k)-1]$ are too high to come before $s_k$ in $P_j^2$; see Figure \ref{fig:PricesGreater4}.  

\item $\ell(s_k)$ is visited by $j$ prior to $s_{k-1}$: by Lemma~\ref{lem:path-prop-1}, since $r(s_k) \not\in P_j$, and $s_k\in P_j^2$, it must be that $\ell(s_k) \in P_j$ and is visited by $P_j$ before $s_k$, and therefore also before $s_{k-1}$. 

\item The price at $\ell(s_k)$ is strictly smaller than that at $s_{k-1}$: if the prices at the two slots were equal, noting that $s_{k-1}<\ell(s_k)$ would imply $s_{k-1}\in P_j^1$, but we know that $s_{k-1}\in P_j^2$.

\item There is a path $Q$ in $\tilde D$ from $\ell(s_k)$ to $s_{k-1}$ to $r(s_k)$ to $s_k$: this follows from the total order on ancestors of $s_k$, because all of these nodes (including $s_{k-1}$ by virtue of it being in $\tilde P_j^2$) are ancestors of $s_k$.

\item The path $Q$ contains an edge $(t', t'')$ where $t' < \ell(s_k)$ and $t'' \ge r(s_k)$.
\end{itemize}
But the final observation yields a contradiction: since the edge $(t', t'')$ goes left to right, it would have to be that $t' = \ell(t'')$. But that can't be, since $\ell(t)$ also has price less than $p_{t''}$ and is further to the right. See Figure \ref{fig:PricesGreater4}. This completes the proof.
\end{proof}

\begin{numberedlemma}{\ref{lem:path-reduction-2}}
The collection of paths $\tilde\realPaths=({\tilde P}_j)$, as defined
above, satisfies the min-work condition. Further, 
a node $t\in \integers_+$ is overloaded under the realized paths
$\tilde \realPaths$ if and only if it is overloaded under the realized
paths $\realPaths$. That is, $\ell_t(\realPaths)\ge B_t$ if and only if
$\ell_t(\tilde\realPaths)\ge B_t$.
\end{numberedlemma}

\begin{proof}
This lemma follows immediately by repeatedly applying Corollary~\ref{lem:short-cut}.

\end{proof}

\subsection{Proof of the stability of service theorem for unit-length jobs}

We are now ready to prove Theorem~\ref{thm:stability-cloud} for the special case of unit-length jobs. Lemmas~\ref{lem:path-reduction-1} and \ref{lem:path-reduction-2} imply that the (random) collection of paths $\tilde\realPaths$ forms a valid instance of the network of servers setting with graph $D$. It remains to argue that the moment generating function of $A_t-B_t$ for every node $t$ is small. Let $\{X_{j,t}\}$ be the fractional assignment given by Lemma~\ref{lem:fractional-assignment} for $\epsilon$ picked in the statement of the theorem.
Recall that $\sum_j q_j X_{j,t}\le (1-\epsilon)B_t$ and $A_t=\sum_j q_j \hat{X}_{j,t}$, where $\hat{X}_{j,t}$ is a Bernoulli random variable with expectation $X_{j,t}$. Then we have:
\begin{align*}
  \EE { e^ {\epsilon (A_t- B_t)}} & \le e^{-\frac 12 \epsilon^2 B_t}\\
& \le \epsilon^{c/2} <\epsilon ^2/3e
\end{align*}
for an appropriate choice of $c$. We can therefore apply Theorem~\ref{thm:Network-main} and the unit-length case of Theorem~\ref{thm:stability-cloud}  follows.

\section{The expected case LP: Proof 
of Lemma~\ref{lem:fractional-assignment}}
\label{sec:LP}

\begin{numberedlemma}{\ref{lem:fractional-assignment}}
{\bf (Fractional assignment lemma)}
Fix any set of potential jobs $J$, their arrival probabilities, and the 
capacities $B_t$ for all $t\in \integers_+$. Then for any $\epsilon \geq 0$,  
$\exists$ nonnegative prices 
$(p_t)_{t\in\integers_+}$ and a fractional assignment $X_{j,t}\in [0,1]$ from 
jobs $j\in J$ to their favorite slots $t\in \favj$,
such that,
\begin{enumerate}
  \item Every job that can afford to pay the price at its favorite slot is 
  fully scheduled: for every $j$ with $p_t(l_j)<v_j$ for $t\in\favj$, we have 
  $\sum_{t\in \favj} X_{j,t}=1$.
	\item The expected allocation at time $t$ is at most $(1-\epsilon)B_t$: 
	$\forall t, \sum_{j\in J, t' \in [t-l_j+1,t]  } q_j X_{j,t'} \leq 
	(1-\epsilon)B_t$. 
	\item The expected social welfare is at least $(1-\epsilon)$ times the 
	optimum: $\sum_{j \in J, t\in \favj } v_j q_j X_{j,t} \geq 
	(1-\epsilon)\opt$. 
\end{enumerate} 
Further, if the distribution is periodic, the prices are also periodic with the 
same period, and can be 
computed efficiently. 
\end{numberedlemma}
\begin{proof}
We begin by writing a linear program for the fractional assignment
problem. The variables in the program correspond to the fractional
assignment of jobs $j$ to slots $t$, $x_{jt}$, with the interpretation
that if job $j$ arrives, it is assigned with probability $x_{jt}$ to
the interval of time $[t,t+l_j-1]$. Fractional assignments must
satisfy two constraints: (1) every job is assigned with total
probability at most $1$; (2) the expected number of jobs assigned to a
slot $t$ is at most $B_t$. Together these constraints ensure that the
optimal solution to the LP satisfies the last two requirements of the
lemma. The prices we select are based on writing a dual for the program.



The linear program and its dual are given below. In the description
below, we assume that the time horizon for the allocation process is
given by $[H]$ for some large $H$. In the case (discussed below)
where jobs are drawn from a periodic distribution, we allow $H$ to go
to infinity. Recall that $W_j$ denotes the window of starting times
for a job; taking the time horizon into account, it is defined as $W_j
= [s_j, \min(d_j,H)-l_j+1]$.

\begin{center} 
\begin{tabular}{l|l}
{\bf Primal LP} & {\bf Dual LP} \\\\
{
$
  \begin{aligned}
   \text{Maximize } \sum_{j, t \in W_j}v_j q_j x_{jt} & \,\,\text{s.t. } &
   \\\\
\sum_{j} \sum_{t'\in [t-l_j+1, t]\cap W_j} q_j x_{jt'} & \le B_t(1-\epsilon) &
\forall t\in [H]\\
\sum_{t\in W_j} x_{jt} & \le 1 & \forall j\in J\\
x_{jt} & \ge 0 & \forall j\in J, t \in [H]
  \end{aligned}
$
}
&
{
$
\begin{aligned}
\text{Minimize } \sum_t \lambda_tB_t(1-\epsilon) & + \sum_j \mu_j
\,\,\text{s.t. } & \\\\
q_j\sum_{t'\in[t,t+l_j-1]}\lambda_{t'} + \mu_j & \geq v_jq_j & \forall
j\in J, t \in W_j  \\
\mu_j  & \geq 0 & \forall j \in J\\
\lambda_t & \geq 0 & \forall t \in [H]
\end{aligned}
$
}
\end{tabular}
\end{center}

\paragraph{Complementary Slackness conditions (CS).} Let $(x^*,
\lambda^*, \mu^*)$ denote the optimal solutions to the Primal and Dual
programs. Define $p_t := \lambda_t^*$, and observe that $p_t(l) =
\sum_{t'\in [t,t+l-1]} \lambda_{t'}^*$. The following complementary
slackness conditions hold:\footnote{The condition corresponding to
  $\lambda_t^*$ is not relevant to the proof of the lemma.}

\begin{enumerate}
\item For all $j\in J, t\in W_j$, either $x_{jt}^*=0$ or $p_t(l_j) +
  \mu_j^*/q_j = v_j$.
\item For all $j\in J$, either $\mu_j^*=0$ or $\sum_{t\in W_j} x_{jt}^*=1$.
\end{enumerate}

The first CS condition, along with the fact that $p_{t'}(l_j) + \mu_j^*/q_j
\ge v_j$ for all $t'\in W_j$ and $j\in J$, implies that if a job is
assigned (partially) to a starting slot $t$, i.e., $x_{jt}^*>0$, then
$p_t(l_j)\le p_{t'}(l_j)$ for all $t'\in W_j$. In other words, jobs
can only be assigned to one or more of their favorite slots.

Furthermore, if for a job $j\in J$, there exists a time $t\in W_j$
with $p_t(l_j)< v_j$, then we must have $\mu_j^*>0$. Thus, the second
CS condition implies that $\sum_{t\in W_j} x_{jt}^*=1$ and the job is
fully scheduled.

This completes the proof of the lemma.
\end{proof}

\paragraph{Periodicity.} We now prove that for periodic instances, we
can efficiently find periodic prices satisfying the conditions of the
lemma.

We begin by defining periodic instances.  A periodic instance with
period $k\in \integers_+$ is given by a core set $J_0$ of potential
jobs and their probabilities. Let the set $J_i$ be obtained by
shifting all the jobs in $J_0$ by $ki$ time units: \[\forall i \in
\integers_+, J_i := \left\{ (s_j+ki,d_j+ki,l_j,v_j) : j \in J_0
\right\} .\] The full set of potential jobs is defined to be
\[ J = \bigcup_{i=0}^\infty J_i. \] The associated probability for
each job is the same as that for the corresponding job in the core set
$J_0$.  Furthermore, supply is also periodic with period $k$: for all
$t=t' \pmod k$, $B_t=B_{t'}$.  A special case is that of
\emph{i.i.d. distributions}, which are simply periodic distributions
with period 1.

We will show next that for periodic instances, the above LP can be
simplified into a compact form. The compact LP assigns jobs in $J_0$
to time slots in $[k]$, with jobs and their windows ``wrapping
around'' the interval $k$. In particular, if $d_j-l_j+1\le k$, we
define the window of a job, as before, to be $\tilde W_j = [s_j,
d_j-l_j+1]$. Otherwise, if $d_j-l_j+1> k$ and the length of the
window, $|W_j|=d_j-l_j-s_j+2$, is smaller than $k$, we define $\tilde
W_j = [s_j, k] \cup [1, s_j+|W_j|-k]$. Finally, if $|W_j|\ge k$, we
define $\tilde W_j = [k]$.

Likewise, for a job $j$ scheduled fractionally at a time $t\in [k]$,
if $t+l_j-1>k$, the job ``wraps around'' the interval $[k]$
(potentially multiple times), and places
load on slots $t'\in [0, t+l_j-1-k]$. Specifically, a job with $l_j\ge k$
places a load of $\lfloor \frac{l_j}{k}\rfloor$ (times its fractional
allocation) on {\em every} slot in $[k]$, and an extra unit of load on
$t'\in [k]$ such that $(t'-t) \bmod k$ is less than or equal to $(l_j-1) \bmod k$.
For example, if a job of length $5$ in a setting with period $k=3$
starts at time slot $1$, it places a load of $2$ units on each of
slots $1$ and $2$, and a load of $1$ unit on slot $3$; if instead it
starts at time slot $3$, then it places a load of $2$ units on each of
slots $3$ and $1$, and a load of $1$ unit on slot $2$.

Accordingly, we obtain the following LP. To understand the capacity
constraint in this LP, consider the example of a setting with period
$3$, and a job $j$ of length $5$ with window $[1,3]$. Then, the load
placed by this job on slot $1$ is $2q_j x_{j1} + q_jx_{j2} + 2q_jx_{j3}$.
\begin{align*}
   \text{Maximize } \,\,\,\frac Hk\sum_{j\in J_0, t \in \tilde W_j}v_j q_j x_{jt}
   \,\,\text{s.t. } & &
   \\\\
\sum_{j\in J_0} \sum_{t'\in [k]\cap \tilde W_j} q_j \left\lfloor
\frac{l_j}{k}\right\rfloor x_{jt'} +
\sum_{j\in J_0} \sum_{\substack{t'\in \tilde W_j : \\  (t-t')\bmod k < l_j \bmod k}} q_j x_{jt'} & \le B_t(1-\epsilon) &
\forall t\in [k]\\
\sum_{t\in \tilde W_j} x_{jt} & \le 1 & \forall j\in J_0\\
x_{jt} & \ge 0 & \forall j\in J_0, t \in [k]
\end{align*}

Let $OPT_P$ and $OPT_A$ be the optimal values of the periodic and
aperiodic LPs given above, respectively. We now show that as $H$ goes
to infinity, $OPT_P \geq OPT_A$, and therefore using $OPT_P$ as a
benchmark can only make our approximation factor worse than what it
should actually be. Observe also that the dual of the periodic LP has
a price variable for every slot $t\in [k]$ --- the interpretation is
that prices repeat with periodicity $k$ and $p_{t'}=p_t$ for all $t'=t
\pmod k$.
\begin{claim}
$OPT_P \geq OPT_A$.
\end{claim}
\begin{proof}
For simplicity we will assume that $H$ is a multiple of $k$.
For a job $j\in J_0$ and $j'\in J_i$ for some $i\in\integers_+$, we
say that $j$ is congruent to $j'$, written $j\cong j'$, if $j'$ is
obtained by shifting the arrival time and deadline of $j$ by a
multiple of $k$. For $t\in [k]$, let $S_t = \{t'\in [H]: t'=t\pmod k\}$. Observe that
$|S_t|=H/k$. 

Let $x_{jt}^*$ be the optimal solution to the aperiodic LP for the set
of all jobs $J$.
Consider the following
solution $\forall t \in [k]$ and $j\in J_0$: $x_{jt}^\dagger = 
\frac kH \sum_{j'\cong j} \sum_{t'\in S_t} x_{j't'}^*$.

The value obtained by the solution $x^{\dagger}$ in the periodic LP is
exactly equal to the value obtained by $x^*$ in the aperiodic LP. We
will now prove that $x^{\dagger}$ is feasible for the periodic LP,
which implies the lemma.

Since $x^*$ is feasible for the aperiodic LP, we have for all $j\in
J_0$:
\[
\sum_{t\in \tilde W_j} x_{jt}^{\dagger} = \sum_{t\in \tilde W_j} \frac
kH \sum_{j'\cong j} \sum_{t'\in S_t} x_{j't'}^* = \frac kH \sum_{j'\cong j} \sum_{t\in \tilde W_j} \sum_{t'\in S_t} x_{j't'}^* 
= \frac kH \sum_{j'\cong j} \sum_{t'\in W_{j'}} x_{j't'}^* \le \frac
kH \sum_{j'\cong j} 1 = 1.
\]

Likewise, we have for all $t\in [k]$:
\[
\sum_{t'\in S_t} \sum_{j} \sum_{t''\in [t'-l_j+1, t']\cap W_j} q_j 
x_{jt''}^* \le \frac Hk B_t(1-\epsilon)
\] 
Rearranging the sum, we get,
\[
\sum_{j\in J_0}q_j \frac kH\sum_{j': j'\cong j} \sum_{t'\in S_t} \sum_{t''\in
  [t'-l_j+1, t']\cap W_j} 
x_{j't''}^* \le B_t(1-\epsilon)
\]

Consider the inner sum $\sum_{t'\in S_t} \sum_{t''\in [t'-l_j+1,
  t']\cap W_j} x_{j't''}^*$. If $l_j<k$, this sum is exactly equal to
the sum over $t'\in\tilde W_j$ with $(t-t')\bmod k < l_j \bmod k$ of
$\sum_{t''\in S_{t'}} x_{j't''}^*$. If $l_j>k$, then each $t''$
belongs to the interval $[t'-l_j+1, t']$ for $\left\lfloor \frac
  {l_j}{k}\right\rfloor$ additional different $t'\in S_t$, and so, we get an
additional term of $\sum_{t'\in [k]}\sum_{t''\in S_{t'}} \left\lfloor
  \frac {l_j}{k}\right\rfloor x_{j't''}^*$. Putting these expressions
together, we get,

\[
\sum_{j\in J_0} q_j \frac kH\sum_{j': j'\cong j} \left(
\sum_{\substack{t'\in \tilde W_j : \\  (t-t')\bmod k < l_j \bmod k}}\sum_{t''\in S_{t'}} x_{j't''}^* + \sum_{t'\in [k]}\sum_{t''\in S_{t'}}
\left\lfloor \frac {l_j}{k}\right\rfloor x_{j't''}^* \right)
\le B_t(1-\epsilon),
\]
or,
\[
\sum_{j\in J_0} \sum_{t'\in [k]\cap \tilde W_j} q_j \left\lfloor
\frac{l_j}{k}\right\rfloor x_{jt'}^{\dagger} +
\sum_{j\in J_0} \sum_{t'\in [t-(l_j\bmod k)+1, t]\cap \tilde W_j} q_j x_{jt'}^{\dagger} \le B_t(1-\epsilon)
\]

Thus, the $x_{jt}^\dagger$'s 
form a 
feasible solution to the periodic LP and the lemma holds.
\end{proof}

\end{document}